\documentclass[11pt]{article}

\usepackage[utf8]{inputenc}
\usepackage{amsmath,amssymb,amsthm}
\usepackage{geometry}
\usepackage{hyperref}
\usepackage{graphicx}
\usepackage{xcolor}
\usepackage{setspace}
\usepackage{enumitem}
\usepackage{booktabs}
\usepackage{natbib}
\usepackage{microtype}
\usepackage{tikz,pgfplots}
\pgfplotsset{compat=1.18}

\title{A Model of Artificial Jagged Intelligence}
\author{Joshua S. Gans\thanks{Rotman School of Management, University of Toronto and NBER. Research assistance from Refine.ink and ChatGPT 5.2 and funding from the SSHRC are acknowledged. Thanks to Tom Cunningham for helpful comments. Responsibility for all errors remains my own. All correspondence to joshua.gans@utoronto.ca.}}
\date{\today}

\newtheorem{theorem}{Theorem}
\newtheorem{proposition}{Proposition}
\newtheorem{lemma}{Lemma}
\newtheorem{definition}{Definition}
\newtheorem{assumption}{Assumption}
\newtheorem{corollary}{Corollary}

\newcommand{\E}{\mathbb{E}}
\newcommand{\Var}{\mathrm{Var}}
\newcommand{\R}{\mathbb{R}}
\newcommand{\GP}{\mathcal{GP}}
\newcommand{\1}{\mathbf{1}}
\newcommand{\CV}{\mathrm{CV}}

\begin{document}

\maketitle

\begin{abstract}
\noindent Generative AI systems often display highly uneven performance across tasks that appear ``nearby'': they can be excellent on one prompt and confidently wrong on another with only small changes in wording or context. We call this phenomenon \emph{Artificial Jagged Intelligence} (AJI). This paper develops a tractable economic model of AJI that treats adoption as an information problem: users care about \emph{local} reliability, but typically observe only coarse, global quality signals. In a baseline one-dimensional landscape, truth is a rough Brownian process, and the model ``knows'' scattered points drawn from a Poisson process. The model interpolates optimally, and the local error is measured by posterior variance. We derive an adoption threshold for a blind user, show that experienced errors are amplified by the inspection paradox, and interpret scaling laws as denser coverage that improves average quality without eliminating jaggedness. We then study mastery and calibration: a calibrated user who can condition on local uncertainty enjoys positive expected value even in domains that fail the blind adoption test. Modelling mastery as learning a reliability map via Gaussian process regression yields a learning-rate bound driven by information gain, clarifying when discovering ``where the model works'' is slow. Finally, we study how scaling interacts with discoverability: when calibrated signals and user mastery accelerate the harvesting of scale improvements, and when opacity can make gains from scaling effectively invisible. \textit{Journal of Economic Literature} Classification Numbers: D83, D81, O33, L86.
\medskip

\end{abstract}

\newpage

\section{Introduction}\label{sec:intro}

Generative AI is being integrated into knowledge work with unusual speed. The novelty is not that the technology is imperfect; most technologies (and, indeed, humans) are. The novelty is that the imperfections are often \emph{local} and \emph{opaque}. Users routinely encounter a pattern that is hard to reconcile with an ``average accuracy'' mindset: a model can produce a crisp, correct answer on one prompt and then produce a plausible, confidently incorrect answer on a nearby-looking prompt. Small changes in phrasing, context length, or formatting can flip performance in ways that are difficult to predict ex ante. This unevenness is increasingly salient because it is precisely in high-value settings (legal, medical, policy, finance, research) that users care most about avoiding rare-but-costly failures.

Empirical evidence is consistent with this local, uneven pattern. In an online experiment on professional writing tasks, access to ChatGPT increased average productivity and quality \citep{noyzhang2023experimental}. In a field setting with customer-support agents, a generative-AI assistant raised average productivity with substantial heterogeneity across workers and problem types \citep{brynjolfsson2025generative}. Most directly, a field experiment with management consultants characterises a ``jagged technological frontier'' in which AI substantially improves performance on some realistic tasks but degrades performance on others that appear similar \citep{dellacqua2023jagged}. These findings motivate a theory in which adoption depends not only on mean quality but on the discoverability of local reliability.

This paper offers a simple model for that phenomenon, which has been termed \emph{Artificial Jagged Intelligence} (AJI).\footnote{Andrej Karpathy used the term ``jagged intelligence'' to describe uneven performance across tasks; see \url{https://x.com/karpathy/status/1816531576228053133}. We use the term in that spirit, focusing on how uneven local reliability shapes adoption, calibration, and learning.} The modelling objective is not to reproduce engineering details of any particular system. It is to supply a clean economic environment in which three facts can coexist:
\begin{enumerate}[label=(\roman*)]
\item \textbf{Local heterogeneity:} performance is uneven across a task space, with ``pockets'' of high competence and ``holes'' of high error;
\item \textbf{Opacity:} users do not initially observe where those pockets and holes are; and
\item \textbf{Discoverability matters:} adoption and productivity depend not only on average quality, but on whether users can locate the reliable regions and avoid the unreliable ones.
\end{enumerate}
In this sense, AJI is best viewed as an \emph{information problem} rather than an engineering bug. A model can be very good on average and still be difficult to use productively if local reliability is jagged and hard to infer.

\subsection{Why economists should care}

Economists have long studied technology adoption when quality is uncertain and revealed through use. In classic work on experience goods, the key friction is that buyers learn quality by consuming the good \citep{nelson1970consumer}. In classic work on adverse selection, the key friction is that quality is privately observed, producing market unravelling \citep{akerlof1970lemons}. Generative AI introduces a distinct but related friction: quality uncertainty is not only across products or firms but \emph{within} a product, \emph{across tasks}. The relevant object for a user is not ``the model's accuracy'' but the model's reliability at the task the user happens to face.

This within-tool heterogeneity has immediate economic implications. First, adoption can be excessively conservative when a tool's local reliability is hard to discover: even if the tool is valuable on a nontrivial subset of tasks, users may avoid it if they fear rare catastrophic failures. Second, once adopted, users may misallocate tasks to the model, delegating where the model is unreliable and failing to delegate where it is reliable, which creates productivity losses and trust dynamics. Third, because mastery and interface design affect discoverability, investments in calibration, workflow design, and user training are economically complementary to raw model improvements. This resonates with a broader view of major technological change as requiring complementary innovations and organisational adaptation \citep{bresnahan1995gpt}.

\subsection{A modelling strategy}

The modelling challenge is to generate jagged local performance without introducing unnecessary complexity. Our strategy has two ingredients.

\smallskip
\noindent\textbf{Knowledge points.} We represent what the model ``knows'' as a set of tasks at which it can answer correctly. This is obviously stylised, but it captures a core idea: training, retrieval, and finetuning produce dense competence in some regions and sparse competence in others. We model knowledge points as a Poisson point process with intensity $\lambda$, which is a reduced-form measure of \emph{coverage} and is the natural object to link to AI ``scale''. Importantly, a user does not observe where these points are.

This knowledge-point representation is closely related to the ``spatial'' view of knowledge in \citet{carnehl2025quest}. Our emphasis differs: we treat the primary economic friction as the user's inability to observe local reliability and the resulting role for calibration, mastery, and interface design.

\smallskip
\noindent\textbf{A rough truth landscape.} Between knowledge points, the model must interpolate. We represent the mapping from tasks to correct answers as a rough stochastic process (Brownian motion in the baseline). This makes local extrapolation risky and yields a closed-form posterior distribution between knowledge points. Brownian motion is chosen not because the world is literally Brownian, but because it is the simplest ``rough'' process with a Markov property that makes interpolation analytically transparent.\footnote{Methodologically, the paper imports tools from Gaussian process learning and bandit design \citep{rasmussen2006gaussian,srinivas2010gaussian,vakili2021information}. Thus, we follow the work begun in economics by \cite{callander2011searching}; see \cite{bardhi2026learning} for a review.}

Together, these assumptions define a local mean-squared error at each task. Jaggedness emerges because error is small near knowledge points and large in the middle of gaps, and the gap lengths are random.

\subsection{Key results and intuition}
To fix ideas, consider a simple toy model we call the ``Bridge of Knowledge.'' Imagine the AI's task space as a river, and its knowledge points as pylons supporting a bridge. The model interpolates between pylons using planks that sag under weight: short spans are stiff and safe, while long spans are flexible and dangerous. Suppose the bridge consists of two repeating segments: a short 2-meter span and a long 8-meter span.
An engineer evaluating the bridge might calculate the \textit{average gap} as 5 meters. But a user walking across experiences something different. They spend only 20\% of their time on the safe span and 80\% on the dangerous one, so the \textit{experienced average} is 6.8 meters. This is the \textbf{inspection paradox}: users are statistically overexposed to a model's weaknesses. The holes matter more than their frequency suggests.
This simple arithmetic clarifies the economics of adoption and mastery:
\begin{itemize}[leftmargin=*]
\item \textbf{Blind adoption.} A user who cannot observe local reliability must decide ex ante whether to rely on the tool.
If their safety threshold requires gaps under 6 meters, they will correctly refuse - even though the naive benchmark
average suggests the bridge is safe. Evaluation based on averages misleads because adoption depends on \emph{experienced}
risk, which the inspection paradox amplifies.\footnote{In the model outlined below, a ``maximum tolerated gap'' rule maps into the stakes cutoff $q$: under the Brownian-bridge variance \eqref{eq:bridge-variance}, $\E[\sigma^2\mid X]=X/6$, so blind adoption reduces to $q\ge \E[\sigma^2(x)]=\E[X^\ast]/6$.}
\item \textbf{Scaling.} Doubling the density of pylons shrinks the gaps to 1 and 4 meters. But relative jaggedness remains identical: the user still spends 80\% of their time on the longest segments. Scaling improves the \emph{level} of reliability but not the \emph{shape} of heterogeneity. This helps explain why improvements in headline benchmarks coexist with persistent ``surprising'' failures.
\item \textbf{Calibration.} Now, suppose the user can see where the pylons are. They cross the safe 2-meter segments and abstain from the dangerous 8-meter ones. This transforms a risky gamble into a tool with positive expected value---even when blind adoption would be irrational. Calibration unlocks pockets of competence without paying the full cost of the holes.
\item \textbf{Mastery.} Calibration is a benchmark. In practice, users must \emph{learn} the reliability map through experience. We show that when the task space is high-dimensional, or the reliability landscape is rough, this learning is slow even for sophisticated users. There is a speed limit on mastery.
\item \textbf{Complements and substitutes.} Below the blind-adoption threshold, scaling and calibration are complements: scaling alone does not induce adoption, but calibration unlocks value. Above the threshold, they become substitutes: when the model is already good enough for blind use, the marginal benefit of calibration diminishes.
\end{itemize}
This paper makes four contributions to the nascent economics of foundation models.
First, we provide a \textbf{microfounded model of jagged reliability}. A parsimonious ``knowledge-point'' framework generates local error that varies sharply across tasks. This delivers a concrete object---the posterior variance at each task location---that operationalises the informal notion that AI capabilities are ``jagged'' \citep{dellacqua2023jagged}.
Second, we characterise \textbf{adoption under opacity}. By embedding the inspection paradox in an adoption decision, we derive a simple threshold rule and show formally why average benchmark performance can be a poor guide to adoption when local failures are salient. This contributes to the growing literature on AI evaluation \citep{burnell2023rethink} by highlighting a systematic wedge between measured and experienced reliability.
Third, we derive a \textbf{learning-rate bound on mastery}. Modelling mastery as Gaussian process regression over a latent reliability map, we show that worst-case uncertainty declines at a rate governed by information gain. This formalises when ``learning where the model works'' is slow and complements empirical work on human-AI collaboration \citep{lai2023selective}.\footnote{In this regard, it embeds a feature of technology adoption studied by \cite{rosenberg1982using} called ``learning by using.''}
Fourth, we characterise \textbf{complementarities between scaling and discoverability}. The analysis clarifies when interface and governance designs---similarity cues, uncertainty estimates, abstention mechanisms, provenance---are substitutes for versus complements to improvements in raw model performance. This informs both product design and regulatory strategy.

Motivated by a growing empirical literature on the productivity effects of generative AI in controlled and field settings
\citep{noyzhang2023experimental,brynjolfsson2025generative,dellacqua2023jagged}, the paper contributes at the intersection
of three theoretical literatures. First, it connects generative AI adoption to the economics of quality uncertainty and experience goods \citep{nelson1970consumer,akerlof1970lemons}. The novelty is within-tool heterogeneity: local reliability varies across a task space, so ``average quality'' is not the relevant object for many decisions. Second, it relates to work on learning and selection in technology use \citep{jovanovic1982selection}. In our setting, learning is not about a scalar quality parameter but about a function that maps tasks to reliability. The cost of learning depends on the geometry of the task space and on what similarity structure users can exploit. Third, it fits into the broader perspective that general-purpose technologies require complementary investments and organisational change \citep{bresnahan1995gpt}. AJI emphasises a specific complement: \emph{discoverability} tools (calibration, provenance, similarity cues) and mastery capital (workflows and training) that make selective delegation feasible.

The paper builds on the spatial knowledge framework of \citet{carnehl2025quest} but shifts the object of interest from
``average performance'' to the \emph{distribution of local reliability across tasks} and to the user's induced exposure to the sparsest regions. This generates a distinct set of economic implications: (i) an inspection-paradox wedge between gap-uniform benchmarks and task-uniform experience; (ii) a blind-adoption threshold pinned down by second moments of coverage gaps; and (iii) sharp complementarities between scale, regularity, calibration/abstention interfaces, and mastery investments (workflows, prompt libraries, and exploration).

\begin{table}[t]
\centering
\caption{What AJI adds relative to ``average performance'' views}
\label{tab:positioning_aji}
\small
\begin{tabular}{p{0.33\linewidth}p{0.30\linewidth}p{0.30\linewidth}}
\toprule
 & Benchmark-average view & AJI view \\
\midrule
Object & Mean score on a suite & Reliability as a function of tasks \\
Exposure & Representative benchmark item & Length-biased toward sparse regions \\
Main friction & Average capability & Opacity of where errors concentrate \\
Key wedge & ``Model vs benchmark'' & Benchmark vs \emph{experienced} risk \\
Levers & Scale & Scale + regularity + calibration + mastery \\
\bottomrule
\end{tabular}
\end{table}

Section \ref{sec:model} introduces the baseline AJI model, derives the inspection-paradox wedge in experienced error, and characterises blind adoption. Section \ref{sec:scaling} interprets scaling laws as denser knowledge coverage and explains why mean quality can improve without eliminating jaggedness. Section \ref{sec:complements} studies calibration and regularity, and characterises returns to scale when jaggedness is (partially) discoverable. Section \ref{sec:reasoning} introduces reasoning as an optional reliability-enhancing mode and derives cutoff policies for when to reason. Section \ref{sec:mastery} models mastery as learning a reliability map and derives a learning-rate bound that highlights an abstention trap. Section \ref{sec:extensions} discusses extensions and applications.

\section{The AJI Model and Adoption}\label{sec:model}

To understand the economic implications of Artificial Jagged Intelligence (AJI), we require a model that captures the spatial correlation of knowledge and the resulting variability in performance. We adapt the framework established by \citet{carnehl2025quest} (CS), which provides a tractable model of knowledge representation. We then introduce the critical element defining the AJI context: imperfect information regarding the specific realisation of that knowledge.

This section builds the baseline model and derives two core objects: expected experienced error and a blind-adoption threshold. Along the way, we clarify how the primitives map to the generative-AI context.

\subsection{Task space and similarity}

Let the task space be a one-dimensional domain $\mathcal{Z}\subset \R$. A user draws a task $x\in \mathcal{Z}$ and wants the correct scalar answer $Y(x)\in \R$. In practice, tasks are questions, prompts, documents, and contexts that live in a high-dimensional space. The one-dimensional restriction is for closed-form results.\footnote{The one-dimensional restriction is used for closed-form interpolation risk (Brownian bridge variance). The underlying
inspection-paradox mechanism is not inherently one-dimensional: in $d$ dimensions, Poisson coverage induces a Voronoi
tessellation and task-uniform sampling size-biases cell volumes (Lemma~\ref{lem:voronoi_size_bias}). Translating cell
geometry into $\sigma^2(x)$ depends on the kernel/covariance structure, but the core wedge between benchmark weighting
and user experience persists.} It can be interpreted as either (i) an analysis along a one-dimensional manifold of tasks within a workflow, or (ii) a reduced-form representation of distance induced by an embedding in which we restrict attention to variation along a single salient direction.\footnote{Section \ref{sec:mastery} returns to the role of dimension and smoothness when users learn a reliability map.}

The fundamental assumption is that answers to different questions are related: knowing the answer to one question provides information about the answers to similar questions. We model this relationship with a stochastic process.

\begin{assumption}[Brownian landscape]\label{ass:brownian}
The truth function $Y: \mathcal{Z} \to \R$ is a (driftless) Brownian motion in the coordinate $x$.\footnote{Formally, for $x_1<x_2$, $Y(x_2)-Y(x_1)\sim\mathcal{N}(0, x_2-x_1)$ with independent increments. See \citet{karatzas1991brownian}.}
\end{assumption}

\noindent A Brownian motion (or Wiener process) is a continuous stochastic process where the increments are independent and normally distributed. Specifically, the difference in answers between two questions, $Y(x_2) - Y(x_1)$, is normally distributed with a variance equal to the distance $|x_2 - x_1|$. This captures two essential features: (i) \textit{Local Correlation}: answers to nearby questions are highly correlated; and (ii) \textit{Accumulating Uncertainty}: the potential difference between answers grows as the questions become more dissimilar. Assuming a Brownian motion is deliberately ``rough'': it is continuous but nowhere differentiable. This captures a central AJI intuition: even small ``moves'' in task space can produce large changes in the correct output. The Brownian assumption is also analytically convenient: conditional on endpoints, the process is a Brownian bridge with a simple variance formula. The goal is not literal realism but a tractable benchmark that makes local interpolation risky in a transparent way.

\subsection{Data and AI prediction}

The artificial intelligence (AI) model utilises a set of data (i.e., points in the task and answer spaces) under the following assumptions.

\begin{assumption}[Knowledge points]\label{ass:poisson}
The set of knowledge points (or data) $\{x_i\}\subset \mathcal{Z}$ is distributed as a homogeneous Poisson point process with intensity $\lambda>0$. Conditional on the set of points, the model observes $\{(x_i, Y(x_i))\}$.
\end{assumption}

\noindent When presented with a novel question $x \notin \{x_i\}$, the AI (aka the LLM) generates a conjecture based on the set of knowledge points. If $x$ falls between two adjacent knowledge points $x_i$ and $x_{i+1}$ (a gap of length $X_i = x_{i+1} - x_i$), the AI interpolates. The AI's best prediction is the conditional mean, $\hat Y(x)$, which is a linear interpolation between $Y(x_i)$ and $Y(x_{i+1})$. The reliability of this prediction is captured by the conditional variance $\sigma^2(x)$. The Poisson assumption is a reduced-form representation of irregular coverage. In the AI context, $\lambda$ can be interpreted as the density of training support, memorised exemplars, retrieval hits, or other mechanisms that make the model reliable in some neighbourhoods of the task space. We focus on a stationary benchmark so that a single parameter captures coverage; Section \ref{sec:scaling} interprets scaling as an increase in $\lambda$.

The model produces the posterior mean prediction:
\[
\hat Y(x) \equiv \E\!\left[ Y(x)\mid \{(x_i, Y(x_i))\}_{i} \right].
\]
Under Assumption \ref{ass:brownian}, for a task $x$ lying between two consecutive knowledge points $x_i < x < x_{i+1}$, the conditional distribution of $Y(x)$ is a Brownian bridge. The posterior mean is a linear interpolation, and the posterior variance is
\begin{equation}\label{eq:bridge-variance}
\sigma^2(x) \equiv \Var\!\left(Y(x)\mid \{(x_i, Y(x_i))\}_{i}\right)
= \frac{(x-x_i)(x_{i+1}-x)}{x_{i+1}-x_i}.
\end{equation}
In the baseline, the AI is Bayes-optimal given its information: it outputs the best mean-squared prediction. Then $\sigma^2(x)$ is not an arbitrary ``error rate'' but a well-defined object: the irreducible uncertainty at $x$ given what the model knows. 

\subsection{User payoffs and information}

A user can either (i) use the AI output or (ii) take an outside option. We normalise payoffs so that, conditional on using the AI at task $x$, expected utility equals a benefit term minus a loss proportional to mean-squared error:
\begin{equation}\label{eq:utility}
U(x) = 1 - \frac{\sigma^2(x)}{q},
\end{equation}
and the outside option yields $0$. \cite{carnehl2025quest} provide a microfoundation for this. Let a correct answer be worth $B$, the outside option be worth $B_0$, and let mean-squared error reduce payoff by $c\,\sigma^2(x)$. Then \eqref{eq:utility} is an affine normalisation of $B-B_0-c\sigma^2(x)$ and $q \equiv (B-B_0)/c$ is a \emph{stakes} parameter. Higher $q$ corresponds to low stakes or easy verification; lower $q$ corresponds to high-stakes decisions or costly verification.\footnote{The parameter, $q$, can also be seen as a user-specific proxy for judgment (that is, knowledge of payoffs and the costs of errors); see \cite{agrawal2025economics}.}

The key assumption made here is that AI users do not observe $\sigma^2(x)$ at the task they face. They may observe global benchmarks or anecdotal evidence, but not the local reliability map. This motivates the blind adoption case below. This captures the core challenge of AJI: the user cannot ascertain the quality of the AI's answer to a specific question $x$ because they do not know where the nearby knowledge gaps are. 

\subsection{The inspection paradox and experienced error}

Let $X$ denote the distance between consecutive knowledge points. Recall that, under Assumption \ref{ass:poisson}, $X\sim \text{Exponential}(\lambda)$ with $\E[X]=1/\lambda$. By contrast, a user's task is (approximately) uniformly located in the domain. Conditional on a realised partition into gaps, a randomly drawn task is more likely to land in longer gaps. This is the inspection paradox.

\begin{proposition}[Inspection paradox for gaps]\label{prop:inspection_paradox}
Let $X^*$ be the length of the gap containing a uniformly drawn task location. Then $X^*$ has the length-biased density
\[
f_{X^*}(x) = \lambda^2 x e^{-\lambda x}, \qquad x\ge 0,
\]
i.e.\ $X^*\sim \text{Gamma}(2,\lambda)$ and $\E[X^*]=2/\lambda$.
\end{proposition}

\begin{proof}
For a stationary Poisson process, gap lengths have density $f_X(x)=\lambda e^{-\lambda x}$ and mean $\E[X]=1/\lambda$. A uniform location lands in a gap of length $x$ with probability proportional to $x f_X(x)$, yielding
\[
f_{X^*}(x)=\frac{x f_X(x)}{\E[X]}=\frac{x(\lambda e^{-\lambda x})}{1/\lambda}=\lambda^2 x e^{-\lambda x}.
\]
This is Gamma$(2,\lambda)$ with mean $2/\lambda$.
\end{proof}

\noindent The typical gap is $1/\lambda$, but the typical task experiences a gap of $2/\lambda$ because longer gaps occupy more of the space. This is a precise sense in which users are ``more likely to be where the model is extrapolating'' than a naive calculation would suggest. This means that the experienced gap length is double the average gap length. This captures the penalty imposed by jaggedness: irregularity increases the average error experienced by users because they disproportionately encounter the gaps where the LLM performs poorly.

Next, conditional on a gap of length $X$, the location of questions within the gap is uniform. Using \eqref{eq:bridge-variance}, one obtains a clean conditional mean. To see this, let $d$ denote distance from the left endpoint; $d\sim \text{Uniform}(0,X)$. Then $\sigma^2 = d(X-d)/X$. Hence
\[
\E[\sigma^2\mid X] = \frac{1}{X}\E[d(X-d)] = \frac{1}{X}\left(X\E[d]-\E[d^2]\right)
= \frac{1}{X}\left(X\cdot \frac{X}{2}-\frac{X^2}{3}\right)=\frac{X}{6}.
\]
\noindent Combining Proposition \ref{prop:inspection_paradox} with this calculation yields the expected experienced variance.
\begin{equation}\label{eq:variance}
\E[\sigma^2(x)] = \E\!\left[\E[\sigma^2(x)\mid X^*]\right] = \E\!\left[\frac{X^*}{6}\right]
= \frac{1}{6}\cdot \frac{2}{\lambda}=\frac{1}{3\lambda}.
\end{equation}
Expected error falls with coverage $\lambda$, but it is larger than a naive calculation using $\E[X]=1/\lambda$ would imply. Treating the relevant gap as ``typical'' would give $\E[\sigma^2]=1/(6\lambda)$. The inspection paradox doubles exposure to long gaps and doubles expected variance to $1/(3\lambda)$. This distinction matters when adoption decisions are based on the task distribution users actually face.

Many public evaluations report average accuracy on curated benchmark sets. In the baseline model, the relevant average is taken over the \emph{task arrival} distribution. That average weights ``gaps'' in coverage by their length (Proposition \ref{prop:inspection_paradox}), not uniformly. Put differently, the same underlying competence can look better or worse depending on whether evaluation samples tasks uniformly in the space or disproportionately from regions with dense training support. This observation motivates evaluation protocols that report tail-risk measures (e.g.\ worst-case by cluster, or performance conditional on distance from training support) alongside means when the adoption decision is driven by rare costly failures.\footnote{The Poisson/length-bias calculation is exact for a stationary Poisson process on $\mathbb{R}$ when the random task is chosen as one of the tasks that occur, picked ``uniformly'' from among all tasks. On a large bounded domain with tasks arriving uniformly, the same formulas hold approximately for points far from the boundary, and the approximation becomes exact as the domain size increases.}

\subsection{Adoption by a blind user}

A \emph{blind} user must decide whether to rely on the model \emph{ex ante}, without observing $\sigma^2(x)$ at the realised task. This captures subscription adoption, organisational ``use it by default'' policies, or settings where task-level verification is costly.

\begin{assumption}[Blind adoption]\label{ass:blind}
The user adopts the AI if $\E[U(x)]\ge 0$; otherwise, they take the outside option.
\end{assumption}

\noindent With this assumption, we can approve a key result:

\begin{theorem}[Adoption threshold]\label{thm:adoption}
Under Assumptions \ref{ass:brownian}--\ref{ass:blind}, blind adoption occurs if and only if
\[
q \ge \frac{1}{3\lambda}.
\]
Equivalently, defining the \emph{reliability ratio} $R \equiv 3\lambda q$, adoption occurs if and only if $R\ge 1$.
\end{theorem}

\begin{proof}
From \eqref{eq:utility} and \eqref{eq:variance},
\[
\E[U(x)] = 1 - \frac{\E[\sigma^2(x)]}{q} = 1 - \frac{1}{3\lambda q}.
\]
This is nonnegative if and only if $q \ge 1/(3\lambda)$, i.e.\ $R\ge 1$.
\end{proof}

\noindent Under blind adoption, the user commits to relying on the model \emph{before} learning where the realised task $x$ sits in the jagged landscape. Hence, conditional on adoption, the per-task (net) rate of return from using the AI is simply $\bar U_B(\lambda,q)\equiv \E[U(x)]=1-\frac{1}{3\lambda q}=1-\frac{1}{R}$. $R (\equiv 3\lambda q)$ will be a key parameter in what follows. Theorem~\ref{thm:adoption} is equivalently the requirement $\bar U_B(\lambda,q)\ge 0$ or $R \ge 1$.

\subsection{Calibration: a perfectly informed benchmark}\label{sec:calibration}

The blind-adoption threshold in Theorem~\ref{thm:adoption} is driven by \emph{average} expected utility.
But AJI is fundamentally a \emph{local} phenomenon: some tasks lie near knowledge points where the model’s
posterior variance $\sigma^2(x)$ is small, while others fall in the middle of long knowledge gaps where
$\sigma^2(x)$ is large. Many real systems and workflows attempt to expose (imperfect) task-level signals of
local reliability (confidence heuristics, abstention, retrieval provenance, cross-checking tools). To isolate the
economic logic, this subsection introduces a benchmark of \emph{perfect calibration}: the user observes the
realised $\sigma^2(x)$ before deciding whether to rely on the AI for the task at hand.

A perfectly calibrated user compares the AI to the outside option \emph{task by task}. With the outside option
normalised to $0$, the user relies on the AI if and only if $U(x)\ge 0$, equivalently $\sigma^2(x)\le q$.
Thus calibrated per-task utility is
\begin{equation}\label{eq:ucal_payoff}
u_C(x)=\max\{U(x),0\}=\left(1-\frac{\sigma^2(x)}{q}\right)_+,
\qquad (a)_+\equiv \max\{a,0\}.
\end{equation}
The corresponding benchmark value of access to the AI is
\begin{equation}\label{eq:ucal_def}
U_C(\lambda,q)\;\equiv\;\E[u_C(x)]
=\E\!\left[\left(1-\frac{\sigma^2(x)}{q}\right)_+\right].
\end{equation}

\noindent A key simplification in the Brownian--Poisson baseline is that calibrated value depends on $(\lambda,q)$ only
through the reliability ratio $R\equiv 3\lambda q$. Thus, the calibrated benchmark can be derived as follows:

\begin{proposition}[Calibrated expected utility]\label{prop:Ucal}
Under Assumptions~\ref{ass:brownian}--\ref{ass:poisson} (Brownian--Poisson baseline), calibrated expected utility depends on $(\lambda,q)$ only through $R$, i.e.\ $U_C(\lambda,q)=U_C(R)$. Moreover, for every $R>0$, $U_C(R)>0$ and $U_C$ is strictly increasing with
\[
U_C'(R)=\frac{1}{R^2}\E\!\left[3\lambda\sigma^2(x)\,\1\{3\lambda\sigma^2(x)<R\}\right]
\in\left(0,\frac{1}{R^2}\right), \qquad R>0.
\]
Finally, in the Brownian--Poisson baseline $U_C$ admits the integral representation
\[
U_C(R)=\int_0^1 \left[\,1-\frac{6t(1-t)}{R}
+\left(1+\frac{6t(1-t)}{R}\right)\exp\!\left(-\frac{R}{3t(1-t)}\right)\right]dt.
\]
\end{proposition}

\begin{proof}
Write calibrated utility as $\E[(1-\sigma^2(x)/q)_+]$ and substitute $R=3\lambda q$. For convenience in what follows, define the \emph{normalised variance}
\begin{equation}\label{eq:Z_def}
Z\;\equiv\;3\lambda \sigma^2(x).
\end{equation}
so that
\[
U_C(\lambda,q)=\E\!\left[\left(1-\frac{Z}{R}\right)_+\right].
\]
By \eqref{eq:variance}, $\E[Z]=3\lambda\E[\sigma^2(x)]=1$. 

Under the Brownian--Poisson baseline, for a uniformly drawn task, let $X^*$ be the length-biased gap length
(Proposition~\ref{prop:inspection_paradox}) and let $t\sim\mathrm{Uniform}(0,1)$ be the relative location within the
gap, independent of $X^*$. Then $\sigma^2(x)=X^*t(1-t)$ by \eqref{eq:bridge-variance}. Since
$X^*\sim\mathrm{Gamma}(2,\lambda)$, we have $\lambda X^*\sim\mathrm{Gamma}(2,1)$, and therefore
\[
Z=3(\lambda X^*)\,t(1-t)
\]
has a distribution that does not depend on $\lambda$. Hence $U_C$ depends on $(\lambda,q)$ only through $R$.

To see that $U_C(R)>0$ for all $R>0$, note that $Z$ has support arbitrarily close to $0$
(e.g.\ when $t$ lies near an endpoint of a gap, or when the realised gap is short), so $\Pr(Z<R)>0$ and thus
$\E[(1-Z/R)_+]>0$.

For each realised $z\ge 0$, the map $R\mapsto (1-z/R)_+=(1-z/R)\1\{z<R\}$ is nondecreasing. Differentiating
$(1-z/R)\1\{z<R\}$ with respect to $R$ gives $\frac{z}{R^2}\1\{z<R\}$. By dominated convergence,
\[
U_C'(R)=\E\!\left[\frac{Z}{R^2}\,\1\{Z<R\}\right]
=\frac{1}{R^2}\E\!\left[Z\,\1\{Z<R\}\right], \qquad R>0,
\]
which is strictly positive because $\Pr(Z<R)>0$.

For the upper bound, note that $\E[Z]=1$. Indeed,
\[
\E[Z]=3\,\E[\lambda X^*]\,\E[t(1-t)]
=3\cdot 2\cdot \frac{1}{6}=1,
\]
so $\E[Z\1\{Z<R\}]<\E[Z]=1$, implying $U_C'(R)<1/R^2$.

Finally, to obtain the integral representation, fix $t$ and write $a\equiv t(1-t)$. Since
$X^*\sim\mathrm{Gamma}(2,\lambda)$ has density $f_{X^*}(x)=\lambda^2 x e^{-\lambda x}$ and $\sigma^2=X^*a$,
\[
\E\!\left[\left(1-\frac{\sigma^2}{q}\right)_+ \Bigm| t\right]
=\int_{0}^{q/a}\left(1-\frac{ax}{q}\right)\lambda^2 x e^{-\lambda x}\,dx
=1-\frac{2a}{\lambda q}+\left(1+\frac{2a}{\lambda q}\right)\exp\!\left(-\frac{\lambda q}{a}\right).
\]
Substituting $\lambda q=R/3$ and $a=t(1-t)$ and integrating over $t\sim\mathrm{Uniform}(0,1)$ yields
\[
U_C(R)=\int_0^1 \left[\,1-\frac{6t(1-t)}{R}
+\left(1+\frac{6t(1-t)}{R}\right)\exp\!\left(-\frac{R}{3t(1-t)}\right)\right]dt,
\]
as claimed.
\end{proof}

\noindent Blind reliance forces the user to accept the AI output even on tasks where local uncertainty is so high that
$U(x)<0$. This is exactly where jaggedness matters: the distribution of $\sigma^2(x)$ has a high-error tail
generated by long gaps (amplified by the inspection paradox) and by within-gap geometry
($t(1-t)$ peaks in the middle of a gap where $t$, distributed uniformly on $[0,1]$ is the relative location within the gap). The economic cost of being blind is therefore pinned down by the
\emph{upper tail} of $\sigma^2(x)$ above the cutoff $q$.

Recalling that $\bar U_B(\lambda,q)\equiv \E[U(x)]=1-1/R$ is expected utility from \emph{blindly} relying on the AI on all tasks, then perfect calibration weakly dominates blind reliance because $(a)_+\ge a$ pointwise.
Define the (expected) \emph{cost of blindness} as the loss from not being able to screen out locally negative-value
tasks:
\[
\Delta_B(\lambda,q)\;\equiv\;U_C(\lambda,q)-\bar U_B(\lambda,q)
=\E\!\left[\left(\frac{\sigma^2(x)}{q}-1\right)_+\right]
=\frac{1}{q}\E\!\left[(\sigma^2(x)-q)_+\right].
\]
Thus, jaggedness matters for welfare not only through mean error (which drives $\bar U_B$), but through the
probability mass and \emph{overshoot} in the high-variance tail (which drives $\Delta_B$). If we subtract $\bar U_B(\lambda,q)$ from the integral expression in Proposition \ref{prop:Ucal} and use $\int_0^1 t(1-t)\,dt=1/6$, then, in the Brownian--Poisson baseline,
\[
\Delta_B(R)
=\int_0^1\left(1+\frac{6t(1-t)}{R}\right)\exp\!\left(-\frac{R}{3t(1-t)}\right)\,dt.
\]
This integral can be evaluated in closed form:
\[
\Delta_B(R)
=\frac{2}{9}e^{-2R/3}\left[
R\,K_0\!\left(\frac{2R}{3}\right)+(3-R)\,K_1\!\left(\frac{2R}{3}\right)
\right],
\]
where $K_\nu(\cdot)$ is the modified Bessel function of the second kind.\footnote{The function $K_\nu(\cdot)$ is the modified Bessel function of the second kind (also called the \emph{Macdonald function}), i.e.\ the exponentially decaying solution to the modified Bessel equation $z^2 y''+z y'-(z^2+\nu^2)y=0$. For $z>0$ it admits the integral representation
\[
K_\nu(z)=\int_0^\infty e^{-z\cosh u}\cosh(\nu u)\,du,
\]
which implies $K_\nu(z)>0$ and that $K_\nu$ is decreasing in $z$. Its standard asymptotics are $K_\nu(z)\sim \sqrt{\frac{\pi}{2z}}\,e^{-z}\quad (z\to\infty)$, $K_0(z)\sim -\log(z/2)-\gamma$, $K_1(z)\sim \frac{1}{z}\quad (z\downarrow 0)$ where $\gamma$ is the Euler--Mascheroni constant. With argument $z=2R/3$ in the closed form for $\Delta_B(R)$, these properties imply $\Delta_B(R)>0$, $\Delta_B(R)\sim 1/R$ as $R\downarrow 0$, and $\Delta_B(R)\sim \bigl(\sqrt{3\pi}/(3\sqrt{R})\bigr)e^{-4R/3}$ as $R\to\infty$.} Calibration, therefore, converts jaggedness from a \emph{tax} on use into an \emph{option value}: even when the average blind payoff is low, the user can concentrate reliance on locally reliable regions and abstain in the holes. This benchmark will be useful later because both $U_C(R)$ and the blindness wedge $\Delta_B(R)$ depend on primitives $(\lambda,q)$ only through the single index $R$.

Perfect calibration is a deliberately stark benchmark. A natural extension is to let the user observe a noisy signal
$s(x)=\sigma^2(x)+\nu$ (e.g.\ $\nu\sim\mathcal{N}(0,\tau^2)$) and delegate when $\E[\sigma^2(x)\mid s(x)]\le q$.
As $\tau^2\to 0$ this approaches calibrated delegation; as $\tau^2\to\infty$ it collapses to blind adoption.
This provides a reduced-form way to study how noisy uncertainty communication attenuates the welfare gains from
calibration/abstention interfaces.

\section{The Impact of AI Scaling Laws}\label{sec:scaling}

A large empirical literature documents approximate ``scaling laws'' for modern AI systems: more data, compute, and
engineering effort tend to produce systematically better models \citep{kaplan2020scaling}. In this paper, we interpret
``better'' as \emph{denser effective coverage} of a task space. Formally, scaling is an increase in the knowledge-point intensity~$\lambda$.

This section records the basic comparative statics of scaling, clarifies why jaggedness survives, and introduces a
single index that governs adoption and welfare comparisons holding gap shape fixed. Section~\ref{sec:complements}
then characterises rates of return to these margins.

\subsection{Scale as denser coverage}

First, let's define scaling in the context of the model. 

\begin{definition}[Scaling]\label{def:scaling}
Scaling is an exogenous increase in knowledge-point intensity from $\lambda_0$ to $\lambda_S=S\lambda_0$, where
$S>1$ is a scaling factor.
\end{definition}

\noindent Using this definition, the impact of scaling is as follows:

\begin{proposition}[Scaling weakly reduces local posterior variance]\label{prop:scaling_law}
Fix $0<\lambda_0<\lambda_S$. Couple the corresponding Poisson knowledge-point sets by writing a $\lambda_S$-process
as the union of an independent $\lambda_0$-process and an independent $(\lambda_S-\lambda_0)$-process. Let
$\sigma_0^2(x)$ and $\sigma_S^2(x)$ denote the posterior variances of $Y(x)$ under the two knowledge sets. Then for
every task $x$,
\[
\sigma_S^2(x)\le \sigma_0^2(x)\qquad \text{almost surely}.
\]
\end{proposition}

\begin{proof}
Fix a realisation of the coupled knowledge sets and a task $x$.
Let $u_0<x<v_0$ be the consecutive knowledge points bracketing $x$ in the baseline set, and
let $u_S<x<v_S$ be the consecutive knowledge points bracketing $x$ in the scaled set.
Under the coupling, the scaled set contains the baseline set, so the bracket can only shrink: $u_0\le u_S<x<v_S\le v_0$.

Under Assumption~\ref{ass:brownian}, the posterior variance at $x$ depends only on the two bracketing anchors and is given by
the Brownian-bridge formula \eqref{eq:bridge-variance}. Hence
\[
\sigma_0^2(x)=\frac{(x-u_0)(v_0-x)}{v_0-u_0},
\qquad
\sigma_S^2(x)=\frac{(x-u_S)(v_S-x)}{v_S-u_S}.
\]
If $(u_S,v_S)=(u_0,v_0)$ there is nothing to show.
Otherwise, the scaled set introduces at least one additional anchor in $(u_0,v_0)$, which replaces one endpoint of the
bracketing interval by a point closer to $x$. A direct comparison shows that the bridge variance weakly decreases when an
endpoint moves inward: for $u<x<v'<v$,
\[
\frac{(x-u)(v'-x)}{v'-u}\le \frac{(x-u)(v-x)}{v-u},
\]
and symmetrically when the left endpoint moves rightward.
Therefore $\sigma_S^2(x)\le \sigma_0^2(x)$ for every $x$ on this realisation, and hence almost surely.
\end{proof}

\noindent In applications, increasing $\lambda$ can stand for more parameters, more training data, improved retrieval,
domain fine-tuning, or representation learning that makes tasks effectively ``closer.'' The model is agnostic about
mechanism: all improvements that densify effective coverage load into~$\lambda$.

\subsection{Benchmarks, adoption, and a sufficient index}

In the Brownian--Poisson baseline, expected posterior variance declines linearly in scale:
\[
\E[\sigma^2(x)] = \frac{1}{3\lambda_S}=\frac{1}{3S\lambda_0}.
\]
This captures the familiar ``smooth'' benchmark improvement from scaling. But users are paid on the task they face, not
on average error: adoption and welfare hinge on \emph{local} reliability relative to stakes. In the model, the relevant
comparison is whether the realised local variance $\sigma^2(x)$ is small relative to the tolerance for error~$q$.

A convenient summary of this tradeoff is the composite index
\begin{equation}\label{eq:R_def_returns}
R \;\equiv\; 3\lambda q.
\end{equation}
Two features motivate $R$. First, it aggregates a \emph{supply-side} object ($\lambda$, effective coverage) and a
\emph{demand-side} object ($q$, stakes/verification difficulty) into a single ``reliability'' index: holding $q$ fixed,
scaling by a factor $S$ multiplies $R$ by $S$, while holding $\lambda$ fixed, moving to a higher-stakes environment
reduces $R$ proportionally. Second, in the baseline
\[
R \;=\; \frac{q}{\E[\sigma^2(x)]},
\]
so $R$ can be read as ``stakes measured in units of typical uncertainty.'' As a result, comparative statics in $R$
translate directly into (i) returns to scaling and (ii) cross-domain differences in stakes.

The reason $R$ is sufficient (holding the \emph{shape} of gaps fixed) is a homogeneity property: scaling shrinks gap
lengths proportionally, so local uncertainty rescales in the same way. The next lemma records this invariance in a form
useful for welfare and use decisions, both of which depend on the distribution of $\sigma^2(x)/q$.

\begin{lemma}[Scale invariance of normalised variance]\label{lem:scale_invariance}
Suppose gap lengths admit a scale representation $X=Y/\lambda$ with $Y$ independent of $\lambda$ (i.e.\ scaling changes
intensity but not the shape of the gap-length distribution). Then the distribution of the \emph{normalised} local
variance $3\lambda\sigma^2(x)$ is independent of $\lambda$ (though it may depend on the shape of $Y$). Consequently,
for a fixed gap shape, any welfare or use object that depends on $(\lambda,q)$ only through the distribution of
$\sigma^2(x)/q$ depends on $(\lambda,q)$ only through $R=3\lambda q$.
\end{lemma}

\begin{proof}
Let $X^*$ be the length-biased gap length and let $t\sim\mathrm{Uniform}(0,1)$ denote the relative location within the
gap. Under the Brownian bridge formula, $\sigma^2(x)=X^*t(1-t)$. If $X=Y/\lambda$ with $Y$ independent of $\lambda$,
then $X^*=Y^*/\lambda$, where $Y^*$ is the length-biased version of $Y$, so $\lambda X^*=Y^*$ is independent of
$\lambda$. Hence
\[
3\lambda\sigma^2(x)=3(\lambda X^*)\,t(1-t)=3Y^*\,t(1-t),
\]
which is independent of $\lambda$ because $Y^*$ and $t$ are independent of $\lambda$.

Finally, note that $\sigma^2(x)/q = (3\lambda\sigma^2(x))/R$. Since the distribution of $3\lambda\sigma^2(x)$ is
$\lambda$-invariant, any expression formed by integrating a function of $\sigma^2(x)/q$ depends on $(\lambda,q)$ only
through $R$.
\end{proof}

\subsection{Why jaggedness survives scaling}

Although scale reduces the \emph{level} of uncertainty, it need not improve the \emph{shape} of coverage. Under the
Poisson benchmark, relative dispersion is scale-invariant:

\begin{proposition}[Jaggedness is scale-invariant]\label{prop:cv}
Under Assumption~\ref{ass:poisson}, the coefficient of variation of gap lengths is constant:
\[
\text{CV}(X) = \frac{\sqrt{\Var(X)}}{\E[X]} = 1,
\]
independent of $\lambda$.
\end{proposition}

\begin{proof}
If $X\sim \mathrm{Exponential}(\lambda)$ then $\E[X]=1/\lambda$ and $\Var(X)=1/\lambda^2$, so $\mathrm{CV}(X)=1$.
\end{proof}

\noindent In particular, the inspection-paradox penalty $\E[X^*]/\E[X]=2$ is scale-invariant. Scaling shifts the
distribution of gaps inward but does not change its relative dispersion, so long gaps remain long relative to the mean, and user exposure to long gaps is not proportionally reduced. This helps rationalise why benchmark improvements may fail to translate one-for-one into reductions in ``surprising failure'' rates in high-stakes settings.

The Poisson benchmark is intentionally stationary. In practice, coverage is heterogeneous; a natural extension is a spatially varying intensity $\lambda(x)$, which delivers persistent hard regions even at a large aggregate scale. The stationary benchmark isolates the jaggedness generated by random coverage and length-biased exposure before adding systematic heterogeneity.

\section{Rates of Return to Scale, Calibration, and Regularity}\label{sec:complements}

This section asks a deliberately practical question: if an AI provider (or a deploying organisation) has scarce
engineering attention, where should it go?  Public discussions of ``AI progress'' often collapse very different
activities into a single narrative of \emph{scale}.  AJI sharpens the distinctions.   

Because payoffs are realised task-by-task, not on benchmark averages, it is useful to separate the three investment
margins that are often conflated in practice:
\begin{enumerate}[label=(\roman*),leftmargin=1.25em]
\item \textbf{Scale} (capability): increase $\lambda$, which compresses posterior uncertainty but does not eliminate
holes.
\item \textbf{Regularity} (coverage shape): reduce the dispersion of knowledge gaps, shrinking exposure to rare but
large gaps (and the inspection-paradox penalty).
\item \textbf{Calibration} (information/interface): expose a signal of local reliability (here, $\sigma^2(x)$) and allow
abstention (Section~\ref{sec:calibration}).
\end{enumerate}
In the model, user welfare is realised locally, on the task actually faced, so the marginal payoff to any improvement depends on (i) the \emph{distribution} of local reliability, (ii) the user's information and decision rule, and (iii) how evaluation benchmarks weight the task space. We, therefore, organise the analysis around three themes:
(i) the drivers of the rate of return to \emph{scaling}, \emph{calibration}, and \emph{regularity};
(ii) when these margins are complements or substitutes; and
(iii) how benchmark-based inference can misstate returns, and how AJI suggests corrected evaluation targets. Throughout, we use the reliability index $R \equiv 3\lambda q,$ which aggregates supply-side capability ($\lambda$, effective coverage density) and demand-side stakes ($q$, verification difficulty).  Holding $q$ fixed, scaling by a factor $S$ multiplies $R$ by $S$; holding $\lambda$ fixed, moving to higher-stakes settings reduces $R$ proportionally.  Under stable gap shape (Section~\ref{sec:scaling}), many welfare and use objects depend on $(\lambda,q)$ only through $R$, allowing clean comparative statics.

\subsection{Returns to scaling}

Scaling increases $\lambda$ and therefore weakly reduces local posterior variance pointwise
(Proposition~\ref{prop:scaling_law}).  The economic return to that statistical improvement depends on two margins:
(i) an \emph{extensive} margin (whether the user adopts at all), and (ii) an \emph{intensive} margin (how much the AI is
used conditional on adoption).

A blind user delegates on all tasks once adopted, and adopts if the expected utility is nonnegative
(Assumption~\ref{ass:blind}).  In the Brownian--Poisson baseline, Theorem~\ref{thm:adoption} gives $U_B(R) \;=\; \max\left\{1-\frac{1}{R},\,0\right\}$. The implied marginal return to scale (measured in $R$-units) is
\begin{equation}\label{eq:MR_scale_blind}
\frac{\partial U_B}{\partial R}
=
\begin{cases}
0, & R<1,\\[4pt]
\frac{1}{R^2}, & R>1,
\end{cases}
\qquad\text{and}\qquad
\frac{\partial U_B}{\partial \lambda}=3q\,\frac{\partial U_B}{\partial R}.
\end{equation}
The kink at $R=1$ is central: below the threshold, marginal scaling has \emph{zero realised return} because the user does not adopt; above the threshold, returns are positive but diminish like $1/R^2$.

Under perfect calibration, the user observes $\sigma^2(x)$ and delegates if and only if $\sigma^2(x)<q$
(Section~\ref{sec:calibration}).  Proposition~\ref{prop:Ucal} implies expected utility depends on $(\lambda,q)$ only through $R$; differentiating (as in Proposition~\ref{prop:Ucal}) yields
\begin{equation}\label{eq:MR_scale_calibrated}
\frac{\partial U_C}{\partial R}
=
\frac{1}{R^2}\E\!\left[3\lambda\sigma^2(x)\,\1\{3\lambda\sigma^2(x)<R\}\right]
\in \left(0,\frac{1}{R^2}\right)
\qquad\text{for all }R>0,
\end{equation}
and $\partial U_C/\partial \lambda = 3q\,\partial U_C/\partial R$.  Economically, calibration turns scaling from a
``kinked'' investment (worthless until adoption) into a smooth investment (positive return for all $R>0$), because the user can exploit low-variance regions even when average reliability is too low for blind adoption.

Calibration also changes \emph{how much} the AI is used.  Define the use share as the probability of delegation on a random task:
\[
s_B(R)=\1\{R\ge 1\}
\qquad\text{(blind: use whenever adopted),}
\]
and, under calibration,
\begin{equation}\label{eq:sC_def}
s_C(R) \equiv \Pr(\sigma^2(x)<q).
\end{equation}
Conditioning on $t$ and using $\lambda X^*\sim\mathrm{Gamma}(2,1)$ gives the one-dimensional representation
\begin{equation}\label{eq:sC_integral}
s_C(R)
=\int_0^1 \left[1-\left(1+\frac{R}{3t(1-t)}\right)\exp\!\left(-\frac{R}{3t(1-t)}\right)\right]dt,
\end{equation}
so $s_C(R)\in(0,1)$ for finite $R$ and is increasing in $R$.  Hence calibration can \emph{increase use} in low-$R$
settings (relative to blind non-adoption) but \emph{reduce use} in high-$R$ settings (relative to blind
``always delegate'') while still raising welfare.

\subsection{Returns to calibration}\label{subsec:returns_calibration}

Calibration is an information/interface margin: it does not (necessarily) change the underlying model's competence, but changes the mapping from capability to welfare by enabling \emph{selective delegation}.  In the model, perfect calibration means the user observes the local posterior variance $\sigma^2(x)$ and can abstain when it is too high. 
In practice, ``investing in calibration'' corresponds to building and deploying a \emph{reliability signal} and a
\emph{policy} that conditions on it. In Section \ref{sec:mastery}, we explore how calibration is achieved by users themselves. Here, we focus on investments the AI provider can make. 

Typical implementation levers include:
\begin{itemize}[leftmargin=1.25em]
\item \textbf{Confidence/risk estimation:} training a separate predictor (or using internal signals such as entropy, log-probabilities, self-consistency, retrieval diagnostics) to estimate the probability of error on a given query.
\item \textbf{Calibration of the signal:} fitting and validating a mapping so that stated confidence aligns with
empirical accuracy (e.g.\ temperature scaling, isotonic regression, Platt scaling, or conformal-style calibration), often stratified by domain and query type.
\item \textbf{Selective prediction and abstention:} setting a threshold and routing low-confidence cases to a human, to a safer tool, to retrieval-first workflows, or to ``ask a clarifying question,'' instead of answering.
\item \textbf{Verification layers:} using a critic/verifier model, tool-based checks, or unit tests to generate a
second opinion that feeds into the risk score and abstention rule.
\item \textbf{Interface and governance:} surfacing uncertainty to users (or enforcing review gates) so that high-stakes actions require higher reliability.
\end{itemize}
These interventions primarily reshape \emph{which tasks are attempted} and \emph{how} the system responds under
uncertainty, rather than uniformly improving raw accuracy.

The natural welfare object is the gain from access to a local reliability signal (here, $\sigma^2(x)$) relative to
blind reliance.  Recall that $\bar U_B(\lambda,q)\equiv \E[U(x)]=1-1/R$ is expected utility from blindly relying on the AI on all tasks (ignoring the adoption truncation) while the \emph{cost of blindness} is:
\begin{equation}\label{eq:blindness_cost}
\Delta_B(\lambda,q)\;\equiv\;U_C(\lambda,q)-\bar U_B(\lambda,q)
=\E\!\left[\left(\frac{\sigma^2(x)}{q}-1\right)_+\right]
=\frac{1}{q}\E\!\left[(\sigma^2(x)-q)_+\right].
\end{equation}
This quantity is the return to perfect calibration: it is exactly the expected loss from being unable to screen out locally negative-value tasks.  It depends not only on the mean of $\sigma^2(x)$ but on the tail probability mass and \emph{overshoot} above the cutoff $q$.  In applied terms, calibration is most valuable when the system sometimes faces ``bad regions'' (high variance) that can be detected and avoided at low cost.

\begin{proposition}[Calibration returns shrink with reliability]\label{prop:calibration_returns}
Fix gap shape (so that $U_C$ depends on $(\lambda,q)$ only through $R=3\lambda q$).  Then $\Delta_B(R)>0$ for all
$R>0$ and $\Delta_B(R)$ is strictly decreasing in $R$ with
\[
\Delta_B'(R)
=
-\frac{1}{R^2}\E\!\left[3\lambda\sigma^2(x)\,\1\{3\lambda\sigma^2(x)>R\}\right]
<0.
\]
\end{proposition}

\begin{proof}
For fixed realised $s\ge 0$, the map $R\mapsto \left(\frac{3\lambda s}{R}-1\right)_+$ is nonincreasing and
differentiable for $R>0$ away from the kink at $R=3\lambda s$, with derivative
$-\frac{3\lambda s}{R^2}\1\{3\lambda s>R\}$.  Dominated convergence yields the derivative formula.  Positivity follows
because $3\lambda\sigma^2(x)$ has support arbitrarily close to $0$ and unbounded above in the baseline, so both
$\Pr(3\lambda\sigma^2(x)<R)$ and $\Pr(3\lambda\sigma^2(x)>R)$ are positive for all $R>0$.
\end{proof}

\noindent Proposition~\ref{prop:calibration_returns} clarifies why calibration is most valuable in marginal domains: when $R$ is small, a nontrivial fraction of tasks are locally negative-value, and screening is valuable; as $R$ grows, those tasks become rarer, and the incremental value of screening falls.  In practice, this predicts that investments in reliable uncertainty estimation, selective abstention, and escalation policies deliver the largest welfare gains in high-stakes or low-coverage environments, whereas in high-$R$ environments, they mainly act as a safety layer with diminishing marginal payoff.

\subsection{Returns to regularity}\label{subsec:returns_regularity}

Scaling raises $\lambda$ and reduces the \emph{mean} gap length between ``anchors'' in task space.  In practice, $\lambda$ can stand for more parameters, more pretraining tokens, more compute, better retrieval, or any mechanism that increases overall effective coverage.  \emph{Regularity} is a distinct provider-side margin: it is about \emph{where} coverage is improved, not just \emph{how much} coverage there is on average.

In the model, jaggedness arises from the uneven spacing of knowledge points: some regions are densely anchored, while others are sparsely anchored.  For LLMs, the analogous unevenness arises because training and post-training pipelines do not allocate effort uniformly across the task space.  Concretely, a provider can increase regularity (reduce the dispersion of effective gaps) through choices such as:
\begin{itemize}[leftmargin=1.25em]
\item \textbf{Data acquisition and curation targeted at sparse regions.}  Collecting or licensing domain data that is
underrepresented (e.g.\ specialised technical writing, enterprise formats, minority languages, niche tooling
workflows), and cleaning it to usable quality.  This does not necessarily increase average tokens much, but it shrinks
the worst holes.
\item \textbf{Reweighting and mixture design.}  Changing training weights so that rare but important domains receive
more gradient mass (or are oversampled) relative to already-well-covered domains.  This reshapes coverage, holding the
mean token budget fixed.
\item \textbf{Active learning / failure-driven data.}  Using evaluation, red-teaming, and product telemetry to identify
high-error clusters and then generating or collecting data in those clusters.  In the model language, this is
``splitting the longest gaps'' rather than uniformly adding points.
\item \textbf{Targeted post-training.}  Domain adapters, specialist experts (e.g.\ mixture-of-experts routing), or
fine-tunes for hard regions can reduce variance dispersion even if average benchmark scores barely move.
\item \textbf{Retrieval and tooling coverage.}  Expanding retrieval indices, tool access, and grounding resources
specifically for sparse regions effectively adds ``anchors'' locally, improving regularity from the user's perspective.
\end{itemize}
These levers differ from scaling because they primarily change the \emph{shape} of local reliability: they reallocate effort from already-dense regions toward the worst holes.

Regularity and calibration address different frictions and correspond to different engineering choices.  \emph{Regularity} is an \emph{ex ante} capability intervention: it changes the underlying reliability landscape by shrinking the worst ``holes'' (high-variance regions) through targeted data, reweighting, specialised components, retrieval coverage, or post-training aimed at specific failure clusters.  By contrast, \emph{calibration} is an \emph{ex post} information and policy intervention: it leaves the underlying landscape largely unchanged, but improves the system's ability to \emph{recognise} when it is likely to be wrong and to condition behaviour on that signal (e.g.\ abstention, escalation to a human, retrieval-first routing, tool-based verification, or asking for clarification). (Section \ref{sec:reasoning} then studies a \emph{costly response} that can be taken conditional on (possibly calibrated) information about baseline unreliability: paying for a ``reasoning'' mode that reduces error further. In our terminology, this is not calibration itself; it is an action enabled by calibration signals.) These margins are, therefore, distinct in several common real-world cases: (i) a provider may keep model weights fixed yet deploy a calibrated risk score and deferral policy, reducing failures by \emph{avoiding} risky queries (calibration without regularity); (ii) a provider may close known holes via targeted fine-tuning or improved domain coverage while leaving user-facing confidence and routing unchanged, reducing failures by making the model \emph{better} in those regions (regularity without calibration); and (iii) retrieval/tooling upgrades can implement either margin depending on whether they primarily add competence in sparse regions (regularity) or primarily generate reliable ``no support found'' signals that trigger deferral (calibration).  In short, regularity changes \emph{where the model is good}, while calibration changes \emph{when the system chooses to act on what it knows}.

Formally, let $X$ denote the gap length between adjacent knowledge points and let $X^*$ denote the \emph{length-biased} gap faced by a uniformly drawn task. Dispersion matters because a randomly encountered task is more likely to fall in a longer gap (the inspection paradox), so the user experiences a systematically different gap distribution than a gap-uniform evaluator.

\begin{proposition}[Experienced mean gap length and variance]\label{prop:experienced_gap_mean} Let $X$ denote the (unbiased) gap length with $\E[X]=1/\lambda$ and coefficient of variation $\CV\equiv \sqrt{\Var(X)}/\E[X]$.  Let $X^*$ denote the length-biased gap length faced by a uniformly drawn task. Then
\[
\E[X^*]=\frac{\E[X^2]}{\E[X]}=\E[X]\,(1+\CV^2)=\frac{1+\CV^2}{\lambda}.
\]
Consequently,
\begin{equation}\label{eq:Esig_regular}
\E[\sigma^2(x)]=\frac{\E[X^*]}{6}=\frac{1+\CV^2}{6\lambda}.
\end{equation}
\end{proposition}

\begin{proof}
Length bias gives $\Pr(X^*\in dx)\propto x\,\Pr(X\in dx)$, hence $\E[X^*]=\E[X^2]/\E[X]$.  Writing
$\E[X^2]=\Var(X)+(\E[X])^2=(\CV^2+1)(\E[X])^2$ yields $\E[X^*]=\E[X](1+\CV^2)$.  Finally,
$\sigma^2(x)=X^*t(1-t)$ with $t\sim\mathrm{Uniform}(0,1)$ implies
$\E[\sigma^2]=\E[X^*]\E[t(1-t)]=\E[X^*]/6$.
\end{proof}

\noindent Equation~\eqref{eq:Esig_regular} separates two provider-relevant primitives:
\begin{itemize}[leftmargin=1.25em]
\item \textbf{Density} ($\lambda$): overall capability/coverage, improved by ``more of everything'' (scale).
\item \textbf{Dispersion} ($\CV$): unevenness of coverage, improved by targeted effort that closes the worst holes
(regularity).
\end{itemize}
Dispersion inflates experienced variance multiplicatively via $(1+\CV^2)$, even when average spacing $1/\lambda$ is the
same.  A convenient implication is the \emph{experienced} reliability ratio,
\begin{equation}\label{eq:R_exp_def}
R^{\mathrm{exp}}(\lambda,q,\CV)\;\equiv\;\frac{q}{\E[\sigma^2(x)]}
=\frac{6\lambda q}{1+\CV^2}
=\frac{2R}{1+\CV^2}.
\end{equation}
For blind reliance (before truncation at $0$),
\begin{equation}\label{eq:UB_regular_general}
\bar U_B(\lambda,q,\CV)=1-\frac{\E[\sigma^2(x)]}{q}=1-\frac{1}{R^{\mathrm{exp}}(\lambda,q,\CV)}.
\end{equation}
Hence, dispersion directly shifts the blind adoption threshold: $\bar U_B\ge 0$ is equivalent to
$R\ge (1+\CV^2)/2$.  Operationally, this says: if a system is ``spiky'' (high dispersion), scaling must run farther
before typical users find blanket delegation worthwhile.

In principle, ``regularity'' is multidimensional: a provider can reshape the entire distribution of coverage gaps by targeting data, retrieval, tools, or post-training toward specific weak regions.  For tractability, we summarise this shape by a single dispersion statistic, the coefficient of variation
\[
\CV \;\equiv\; \frac{\sqrt{\Var(X)}}{\E[X]},
\]
because Proposition~\ref{prop:experienced_gap_mean} implies that \emph{for blind expected utility}, dispersion affects experienced mean variance only through $\CV$:
\[
\E[\sigma^2(x)] \;=\; \frac{1+\CV^2}{6\lambda}.
\]
We then introduce a one-dimensional regularity index $r\in[0,1]$ that monotonically reduces dispersion and is
normalised at two polar cases:
\begin{equation}\label{eq:CV_r_def}
\CV(0)=1 \quad\text{(Poisson-like irregularity)},\qquad \CV(1)=0 \quad\text{(perfectly regular spacing)}.
\end{equation}
To keep expressions transparent, we adopt the linear normalisation
\[
\CV(r)=1-r,
\]
so that $r$ can be read as the ``fraction of Poisson dispersion removed.'' Under this parametrisation,
\eqref{eq:Esig_regular} becomes
\begin{equation}\label{eq:Esig_r}
\E[\sigma^2(x)]=\frac{1+(1-r)^2}{6\lambda}.
\end{equation}
This specification is a simplification that, in effect, views the AI provider as investing in reducing $\CV$ directly and is appropriate for the exercise here.\footnote{It is reasonable for the objects used here because, for blind reliance, welfare depends on coverage shape only through $\E[\sigma^2(x)]$, and
\eqref{eq:Esig_regular} shows that $\CV$ is the relevant summary statistic.  Any differentiable, decreasing map
$r\mapsto \CV(r)$ satisfying \eqref{eq:CV_r_def} would deliver the same qualitative comparative statics; the linear choice fixes units and yields closed-form expressions. What we lose is that  $\CV$ is a second-moment summary.  Two gap distributions can share the same $(\lambda,\CV)$ but have very different
\emph{tails} (e.g.\ rare catastrophic gaps versus more moderate dispersion).  This matters especially under
calibration, where welfare and use depend on the full distribution of $\sigma^2(x)$ through truncation at the cutoff $q$.  Thus $r$ should be interpreted as capturing the component of ``regularity'' that operates through the second moment (equivalently $\E[X^2]$), not as a complete description of tail engineering. Finally, the simplification restricts $r\in[0,1]$ which implies $\CV\in[0,1]$, i.e.\ coverage is no more dispersed than the Poisson benchmark.  This is appropriate if we view Poisson-like spacing as a natural high-dispersion baseline and interpret regularity investments as reducing dispersion.  If some environments feature heavier-tailed gaps ($\CV>1$), the analysis extends by allowing $r<0$ or by using an alternative monotone map (e.g.\ $\CV(r)=1/(1+r)$ for $r>-1$) without changing the economic logic.}

\begin{definition}[Provider investment technology]\label{def:investment_tech}
A provider chooses coverage intensity $\lambda>0$ and regularity $r\in[0,1]$ at cost $C(\lambda,r)$, with
$C_\lambda>0$, $C_r>0$, and $C$ convex in each argument.
\end{definition}

\noindent Convexity is economically natural: the cheapest regularity gains come from fixing obvious, repeatedly observed holes (e.g.\ high-frequency failure clusters); subsequent gains require progressively harder data acquisition, curation, tooling, or specialised training.  Many real interventions move both margins at once (e.g.\ retrieval can increase effective $\lambda$ while also improving regularity by disproportionately helping sparse regions); we separate $(\lambda,r)$ to isolate the comparative statics of ``more coverage'' versus ``more even coverage.''

Regularity and scale are substitutes at the level of the mean in \eqref{eq:Esig_r} (both act through $\lambda/(1+(1-r)^2)$), but their marginal rates differ.

\begin{proposition}[Relative returns: scale vs.\ regularity]\label{prop:MRS_scale_reg}
In the blind model (conditional on adoption), let $\CV(r)$ be differentiable with $\CV'(r)\le 0$.  Then the marginal
rate of substitution between scale and regularity is
\[
\mathrm{MRS}_{\lambda,r}
\;\equiv\;
\frac{\partial \bar U_B/\partial r}{\partial \bar U_B/\partial \lambda}
=
\lambda\,
\frac{-\frac{d}{dr}\bigl(1+\CV(r)^2\bigr)}{1+\CV(r)^2}
=
\lambda\,\frac{-2\CV(r)\CV'(r)}{1+\CV(r)^2}.
\]
Under the linear normalisation $\CV(r)=1-r$, this reduces to
\[
\mathrm{MRS}_{\lambda,r}
=
\frac{2(1-r)\lambda}{1+(1-r)^2},
\]
which is increasing in $\lambda$ and tends to $0$ as $r\to 1$.
\end{proposition}

\begin{proof}
Blind expected utility (before truncation) is $\bar U_B(\lambda,q,r)=1-\E[\sigma^2(x)]/q$.
By Proposition~\ref{prop:experienced_gap_mean}, $\E[\sigma^2(x)]=(1+\CV(r)^2)/(6\lambda)$, so
\[
\bar U_B(\lambda,q,r)=1-\frac{1+\CV(r)^2}{6\lambda q}.
\]
Differentiating yields
$\partial \bar U_B/\partial r= -\frac{1}{6\lambda q}\frac{d}{dr}(1+\CV(r)^2)$
and
$\partial \bar U_B/\partial \lambda=\frac{1+\CV(r)^2}{6q\,\lambda^2}$.
Taking the ratio gives the general expression.  Substituting $\CV(r)=1-r$ gives the closed form and its
comparative statics.
\end{proof}

\noindent Proposition~\ref{prop:MRS_scale_reg} clarifies what the one-dimensional regularity index is doing economically: it tracks how quickly provider effort compresses dispersion (through $\CV'(r)$), and, therefore, how rapidly the experienced second moment $\E[X^2]$ falls relative to further increases in density.  The qualitative implication is robust to the exact parametrisation: as $\lambda$ grows and average gaps are already short, the marginal value of closing the remaining irregular holes rises relative to adding another ``typical'' point, even though both reduce mean variance.

\subsection{How the margins interact: complements and substitutes}

The three margins interact because they address different frictions.  Scale reduces uncertainty everywhere; regularity reduces dispersion and tail exposure; calibration changes behaviour by letting the user condition on local reliability. AJI implies that whether margins are complements or substitutes depends on which margin is binding.

\paragraph{Scale \emph{vs.} calibration: }
Comparing \eqref{eq:MR_scale_blind} and \eqref{eq:MR_scale_calibrated} yields a sharp interaction:
\begin{itemize}[leftmargin=1.25em]
\item For $R<1$, blind users do not adopt and $\partial U_B/\partial R=0$, whereas calibrated users obtain
$\partial U_C/\partial R>0$.  Calibration \emph{complements} scale by making marginal scaling pay off even below the blind adoption threshold.
\item For $R>1$, blind users already delegate everywhere, so marginal scaling is valued at $1/R^2$; under calibration,
$\partial U_C/\partial R<1/R^2$ because some of the negative-value tail has already been avoided.  Conditional on broad adoption, calibration \emph{substitutes} for scale at the margin by attenuating the incremental payoff to further scale.
\end{itemize}
Thus, ``returns to scale'' is not a purely technological slope: it is a behavioural object that depends on what users can observe and how they respond.

\paragraph{Scale \emph{vs.} regularity:}
Holding $q$ fixed, both higher $\lambda$ and higher $r$ raise the experienced reliability ratio
$R^{\mathrm{exp}}=6\lambda q/(1+(1-r)^2)$, so they are substitutes for average (blind) performance.  But because the marginal return to $\lambda$ falls like $1/\lambda^2$ while the marginal return to $r$ falls like $1/\lambda$,
Proposition~\ref{prop:MRS_scale_reg} implies the \emph{relative} return to regularity rises with scale.  This is the formal sense in which ``once you have scaled a lot, closing the remaining holes becomes the high-leverage move.''

\paragraph{Calibration \emph{vs.} regularity: }
Calibration mitigates holes \emph{ex post} by screening; regularity mitigates holes \emph{ex ante} by shrinking the high-variance tail.  The cost-of-blindness identity \eqref{eq:blindness_cost} makes the substitution transparent: since $x\mapsto (x-q)_+$ is convex, any mean-preserving contraction of the distribution of $\sigma^2(x)$ reduces the return to calibration.

\begin{lemma}[Regularity reduces the value of calibration under tail contraction]\label{lem:reg_reduces_calib}
Fix $(\lambda,q)$.  If an increase in regularity transforms $\sigma^2(x)$ into a mean-preserving contraction in the
convex order (i.e.\ it weakly reduces $\E[\phi(\sigma^2(x))]$ for every convex $\phi$ while keeping $\E[\sigma^2(x)]$
fixed), then the cost of blindness $\Delta_B(\lambda,q)$ weakly decreases.
\end{lemma}

\begin{proof}
From \eqref{eq:blindness_cost}, $\Delta_B(\lambda,q)=\frac{1}{q}\E[(\sigma^2(x)-q)_+]$.  The function
$s\mapsto (s-q)_+$ is convex, so under a mean-preserving contraction in the convex order, its expectation weakly
decreases.  Multiplying by $1/q$ preserves the inequality.
\end{proof}

\noindent Lemma~\ref{lem:reg_reduces_calib} formalises a simple idea: calibration and regularity are both tail-risk tools, so improving one lowers the marginal value of the other.  At the same time, calibration can increase the realised return to regularity in low-$R$ domains by preventing the extensive-margin collapse of blind adoption: without calibration, a user may not adopt at all, making even large regularity improvements privately irrelevant.

\subsection{Benchmarks, mismeasured returns, and AJI-adjusted evaluation}

AJI is also a \emph{measurement} problem: common benchmarks often weight tasks or domains in ways that do not match how users encounter the task space.  To put it in the model's terms, benchmark suites often sample closer to a \emph{gap-uniform} (or otherwise
dense-region-weighted) distribution, whereas real-world use samples tasks across the space and therefore induces
length-biased exposure to the sparsest regions. Because users are length-biased toward sparse regions, benchmarks that ignore length bias can misstate both the \emph{level} and the \emph{marginal return} to investments. 

A benchmark regime that effectively weights gaps uniformly (one item per gap/domain) evaluates an ``unbiased'' gap
length $X$ rather than the user's length-biased $X^*$.  Since the within-gap average posterior variance is $X/6$,
\[
\E_{\text{bench}}[\sigma^2]=\frac{\E[X]}{6}=\frac{1}{6\lambda},
\qquad
\E_{\text{user}}[\sigma^2]=\frac{\E[X^*]}{6}=\frac{1+\CV^2}{6\lambda}.
\]
Equivalently, the implied reliability ratios satisfy
\[
\frac{R^{\mathrm{exp}}}{R^{\mathrm{bench}}}
=\frac{q/\E_{\text{user}}[\sigma^2]}{q/\E_{\text{bench}}[\sigma^2]}
=\frac{1}{1+\CV^2}\le 1,
\]
with equality only under perfectly regular spacing ($\CV=0$).  In the Poisson benchmark ($\CV=1$), the wedge is exactly $1/2$: a gap-uniform benchmark is twice as optimistic about reliability as a representative user's experience.

If returns to scaling are inferred from benchmark slopes, AJI generates two distortions:
\begin{enumerate}[label=(\roman*),leftmargin=1.25em]
\item \textbf{Extensive-margin distortion (adoption).} Benchmarks that ignore length bias predict adoption too early: they evaluate at $\E[X]$ while users experience $\E[X^*]$.
\item \textbf{Marginal-return distortion.} Even holding behaviour fixed, benchmark slopes overstate the marginal
improvement from scaling because a gap-uniform benchmark effectively evaluates $q/\E_{\text{bench}}[\sigma^2]$ while user
experience is pinned down by $q/\E_{\text{user}}[\sigma^2]$.
Since $\E_{\text{user}}[\sigma^2]=(1+\CV^2)\E_{\text{bench}}[\sigma^2]$, the experienced reliability ratio is
$R^{\mathrm{exp}}=R^{\mathrm{bench}}/(1+\CV^2)$, so the marginal gain from increasing $\lambda$ is attenuated by the factor
$1/(1+\CV^2)$.
\end{enumerate}
The model suggests an immediate correction: report (or target) the \emph{second raw moment} of gap lengths, $\E[X^2]$ (together with $\E[X]$), since task-uniform use induces length bias and the experienced mean gap satisfies $\E[X^*]=\E[X^2]/\E[X]$.  Operationally, this means either (i) evaluating by sampling tasks uniformly from the task space (which automatically induces length bias), or (ii) adjusting gap-uniform benchmark estimates by an estimated dispersion factor $(1+\CV^2)$ (or its empirical analogue). Appendix \ref{app:benchmarks_aji} reviews commonly used benchmarks for LLMs and highlights the implications for their interpretation from the model as well as suggestions for their adjustment.

To make the investment implication transparent, consider a finite task space of length $L$ with gaps
$(X_i)_{i=1}^n$, $\sum_i X_i=L$.  A uniformly drawn task lands in gap $i$ with probability $X_i/L$ and has within-gap
average variance $X_i/6$, so the experienced mean variance is
\begin{equation}\label{eq:Esig_discrete}
\E[\sigma^2]=\sum_{i=1}^n \frac{X_i}{L}\cdot \frac{X_i}{6}
=\frac{1}{6L}\sum_{i=1}^n X_i^2.
\end{equation}
By contrast, a gap-uniform benchmark effectively evaluates
\[
\E_{\text{gap-unif}}[\sigma^2]=\frac{1}{n}\sum_{i=1}^n \frac{X_i}{6}
=\frac{L}{6n},
\]
which depends only on $n$ and is invariant to dispersion in $(X_i)$.

\begin{proposition}[Where a marginal knowledge point matters]\label{prop:target_long_gaps}
Suppose a provider can add one knowledge point inside a chosen gap of length $X$, splitting it into two gaps of
lengths $\alpha X$ and $(1-\alpha)X$ for some $\alpha\in(0,1)$.  Then the reduction in experienced mean variance
\eqref{eq:Esig_discrete} is
\[
\Delta \E[\sigma^2]
=\frac{1}{6L}\Big(X^2-(\alpha^2+(1-\alpha)^2)X^2\Big)
=\frac{1}{3L}\alpha(1-\alpha)X^2,
\]
which is maximised by splitting the gap in half ($\alpha=1/2$), and is increasing in the original gap length $X$.
In particular, to maximise the experienced-quality gain from a marginal point, the provider should split the
\emph{longest} gap.
\end{proposition}

\begin{proof}
The experienced mean variance is proportional to $\sum_i X_i^2$ by \eqref{eq:Esig_discrete}.  Splitting a gap of
length $X$ replaces $X^2$ with $(\alpha X)^2+((1-\alpha)X)^2=(\alpha^2+(1-\alpha)^2)X^2$, giving the stated reduction. The function $\alpha(1-\alpha)$ is maximised at $\alpha=1/2$ and the reduction scales with $X^2$.\footnote{In their examination of `knowledge deepening,' \cite{carnehl2025quest} prove a similar result regarding adding knowledge points at mid-points.}
\end{proof}

\noindent Proposition~\ref{prop:target_long_gaps} highlights the incentive problem: under a gap-uniform benchmark, adding a point anywhere increases $n$ by one and improves the score by the same amount regardless of whether it closes a catastrophic hole or refines an already dense region.  Under an experienced (usage-weighted) objective, marginal value scales like $X^2$, so closing the worst holes has first-order value.  In short, AJI does not merely change the level of measured performance; it changes which investments look high-return, and therefore which improvements are likely to be supplied.

With calibration, ``usage'' is endogenous: the user abstains precisely in high-variance regions.  A welfare-aligned evaluation should therefore report both (i) value net of abstention, $U_C$, and (ii) coverage, $s_C$ (the use share in \eqref{eq:sC_def}).  Otherwise, it is possible to look good by implicitly refusing to answer precisely where the model is jagged.

\section{Reasoning Modes and the Jagged Landscape}\label{sec:reasoning}

Many AI systems now offer multiple inference modes: a cheap ``default'' model and a more expensive
``reasoning'' (or ``thinking'') mode that uses additional compute (and often latency) to reduce errors.
In a smooth world, this is a standard quality--price tradeoff. In a jagged world, it becomes a \emph{selection} problem: the value of paying for reasoning is highest precisely on tasks that are locally hard, and by Proposition~\ref{prop:inspection_paradox} users are over-exposed to those tasks. This section introduces a simple benchmark model of reasoning and studies how it interacts with jaggedness, the inspection paradox, and the measurement of returns to scale.

\subsection{Setup: two inference modes}

At each task $x$, the user can either abstain (outside option $0$), use a cheap ``fast'' mode ($F$), or invoke a costly ``reasoning'' or ``thinking'' mode ($T$). Fast-mode utility is the baseline payoff from \eqref{eq:utility},
\begin{equation}\label{eq:reasoning_fast}
U_F(x)\;\equiv\;1-\frac{\sigma^2(x)}{q}.
\end{equation}
Reasoning incurs an additional per-task cost $\kappa>0$ (paid in time, money, or attention) and reduces local mean-squared error from $\sigma^2(x)$ to $\sigma_T^2(x)$:
Because payoffs are normalised as in \eqref{eq:utility}, $\kappa$ is measured in the same units: a hypothetical
zero-error reasoning output would yield net utility $1-\kappa$.
\begin{equation}\label{eq:reasoning_slow}
U_T(x)\;\equiv\;1-\frac{\sigma_T^2(x)}{q}-\kappa.
\end{equation}
The per-task value of reasoning relative to fast mode is therefore
\begin{equation}\label{eq:reasoning_increment}
V_T(x)\;\equiv\;U_T(x)-U_F(x)
=\frac{\sigma^2(x)-\sigma_T^2(x)}{q}-\kappa.
\end{equation}
Two features of \eqref{eq:reasoning_fast}--\eqref{eq:reasoning_increment} are economically central.
First, reasoning is an \emph{inference-side} margin: it is paid per task rather than up front, so jaggedness affects the realised ``reasoning bill'' through the task-arrival distribution. Second, the user only gets the full value of reasoning if they can \emph{target} it to tasks where $\sigma^2(x)$ is high; this immediately connects reasoning to calibration (Section~\ref{sec:calibration}).

\subsection{A benchmark model: reasoning as local noisy evidence}

The key modelling choice is how reasoning changes $\sigma^2(x)$. A useful benchmark is to treat reasoning as acquiring an additional, task-specific piece of evidence that is informative about $Y(x)$ but imperfect.
This captures a range of mechanisms (retrieval, tool use, longer search, self-consistency), while keeping the model transparent.

\begin{definition}[Reasoning as a noisy pseudo-anchor]\label{def:reasoning_noisy_anchor}
When reasoning is invoked at task $x$, the model obtains an additional signal
\[
\tilde Y(x)=Y(x)+\varepsilon,\qquad \varepsilon\sim N(0,\sigma_\varepsilon^2),
\]
independent of the baseline information set.
The resulting posterior variance is
\begin{equation}\label{eq:reasoning_variance}
\sigma_T^2(x)
=\left(\frac{1}{\sigma^2(x)}+\frac{1}{\sigma_\varepsilon^2}\right)^{-1}
=\frac{\sigma^2(x)\,\sigma_\varepsilon^2}{\sigma^2(x)+\sigma_\varepsilon^2}.
\end{equation}
\end{definition}

\noindent Definition~\ref{def:reasoning_noisy_anchor} can be interpreted literally (a noisy auxiliary estimate) or as a reduced form for ``extra compute'' that produces an additional independent sample of the answer distribution. The parameter $\sigma_\varepsilon^2$ is a \emph{reasoning noise floor}: even with unlimited baseline uncertainty, reasoning cannot reduce variance below $\sigma_\varepsilon^2$.

\begin{proposition}[How reasoning reshapes local risk]\label{prop:reasoning_properties}
Under Definition~\ref{def:reasoning_noisy_anchor}:
\begin{enumerate}[label=(\roman*)]
\item $\sigma_T^2(x)$ is increasing and concave in $\sigma^2(x)$, with
$\sigma_T^2(x)\le \min\{\sigma^2(x),\sigma_\varepsilon^2\}$ and
$\lim_{\sigma^2(x)\to\infty}\sigma_T^2(x)=\sigma_\varepsilon^2$.
\item The variance reduction
\[
\Delta(x)\;\equiv\;\sigma^2(x)-\sigma_T^2(x)=\frac{\sigma^4(x)}{\sigma^2(x)+\sigma_\varepsilon^2}
\]
is increasing and \emph{convex} in $\sigma^2(x)$.
\item Reasoning therefore makes the landscape ``jagged but capped'': long gaps still generate high baseline variance,
but reasoning limits the realised variance to (approximately) $\sigma_\varepsilon^2$ on those tasks.
\end{enumerate}
\end{proposition}

\begin{proof}
Write $v\equiv \sigma^2(x)$.
From \eqref{eq:reasoning_variance}, $\sigma_T^2(v)=v\sigma_\varepsilon^2/(v+\sigma_\varepsilon^2)$.
Then
\[
\frac{d\sigma_T^2}{dv}=\frac{\sigma_\varepsilon^4}{(v+\sigma_\varepsilon^2)^2}>0,
\qquad
\frac{d^2\sigma_T^2}{dv^2}=-\frac{2\sigma_\varepsilon^4}{(v+\sigma_\varepsilon^2)^3}<0,
\]
establishing monotonicity and concavity, and the bounds and limits follow by inspection.
The reduction is $\Delta(v)=v-\sigma_T^2(v)=v^2/(v+\sigma_\varepsilon^2)$, which satisfies
\[
\frac{d\Delta}{dv}=\frac{v(v+2\sigma_\varepsilon^2)}{(v+\sigma_\varepsilon^2)^2}>0,
\qquad
\frac{d^2\Delta}{dv^2}=\frac{2\sigma_\varepsilon^4}{(v+\sigma_\varepsilon^2)^3}>0,
\]
so $\Delta$ is increasing and convex.
\end{proof}

\noindent If, instead, reasoning reduces variance proportionally, $\sigma_T^2(x)=\theta\sigma^2(x)$ with $\theta\in(0,1)$, then reasoning preserves the \emph{shape} of jaggedness (it rescales the entire landscape) and the inspection-paradox multiplier from Section~\ref{sec:complements} is unchanged. Definition~\ref{def:reasoning_noisy_anchor} is deliberately different: it captures the empirically common pattern that reasoning helps disproportionately in high-uncertainty regions but hits a floor.

\subsection{Optimal reasoning under calibration}

As in Section~\ref{sec:calibration}, consider a calibrated user who observes $\sigma^2(x)$ before choosing whether to use the AI
and which mode to run.\footnote{This is a benchmark: in practice, users have partial signals (confidence scores, self-evaluations,
heuristics) rather than direct access to $\sigma^2(x)$.  The point of the calibrated case is to isolate the economics of the
reasoning option given ideal targeting.}

\begin{proposition}[Cutoff rule for reasoning]\label{prop:reasoning_cutoff}
Under Definition~\ref{def:reasoning_noisy_anchor}, a calibrated user chooses
\[
\max\{0,U_F(x),U_T(x)\}
\]
by a cutoff rule in the baseline variance $v\equiv\sigma^2(x)$.
Define the \emph{fast break-even} point $v_0\equiv q$, the \emph{reasoning-entry threshold} $v_L$ as the unique positive solution to
$U_T(v)=U_F(v)$,
\begin{equation}\label{eq:vL}
v_L \;\equiv\;\frac{\kappa q+\sqrt{\kappa^2q^2+4\kappa q\sigma_\varepsilon^2}}{2},
\end{equation}
and define the \emph{reasoning break-even} point $v_H$ (if it exists) by $U_T(v_H)=0$,
\begin{equation}\label{eq:vH}
v_H\;\equiv\;
\begin{cases}
+\infty, & \sigma_\varepsilon^2\le (1-\kappa)q,\\[4pt]
\dfrac{(1-\kappa)q\,\sigma_\varepsilon^2}{\sigma_\varepsilon^2-(1-\kappa)q}, & \sigma_\varepsilon^2>(1-\kappa)q.
\end{cases}
\end{equation}
Then:
\begin{enumerate}[label=(\roman*)]
\item If $U_T(q)<0$ (equivalently $\kappa>\frac{q}{q+\sigma_\varepsilon^2}$), reasoning is never optimal and the user uses fast mode if and only if $v\le q$,
abstaining otherwise.
\item If $U_T(q)\ge 0$ (equivalently $\kappa\le\frac{q}{q+\sigma_\varepsilon^2}$), then $v_L\le q\le v_H$ and the optimal policy is:
fast mode for $v<v_L$, reasoning for $v\in[v_L,v_H]$, and abstention for $v>v_H$ (with no abstention when $v_H=+\infty$).
\end{enumerate}
\end{proposition}

\begin{proof}
Write $v=\sigma^{2}(x)$ and $\xi\equiv\sigma_\varepsilon^2$.
From \eqref{eq:reasoning_fast}--\eqref{eq:reasoning_variance},
\[
U_F(v)=1-\frac{v}{q},
\qquad
U_T(v)=1-\frac{v\xi}{q(v+\xi)}-\kappa.
\]
Both are decreasing in $v$, and $U_T(v)$ is bounded below by $1-\xi/q-\kappa$. Moreover
\[
U_T(v)-U_F(v)=\frac{v^2}{q(v+\xi)}-\kappa,
\]
which is strictly increasing in $v$ (its derivative is positive for $v>0$). Hence there is a unique crossing point $v_L$ at which $U_T(v)=U_F(v)$, which solves the quadratic $v^2-\kappa q\,v-\kappa q\,\xi=0$ and yields \eqref{eq:vL}.

If $U_T(q)<0$, then $U_T(v)<0$ for all $v\ge q$ (since $U_T$ is decreasing) and $U_T(v)<U_F(v)$ for all $v\le q$
(since $U_T-U_F$ is increasing and negative at $q$), so reasoning is never optimal. If $U_T(q)\ge 0$, then $U_T-U_F$ crosses zero weakly before $q$, so $v_L\le q$, and reasoning dominates fast mode for all $v\ge v_L$. Reasoning dominates abstention if and only if $U_T(v)\ge 0$, which defines $v_H$ in \eqref{eq:vH}. Combining these comparisons yields the stated policy.
\end{proof}

\noindent  Equation \eqref{eq:vL} shows that a calibrated user buys reasoning when baseline variance is high enough that the variance reduction is worth paying the surcharge $\kappa$. Equation \eqref{eq:vH} makes the noise-floor logic transparent: when $\sigma_\varepsilon^2$ is large relative to stakes $q$, even reasoning cannot make the task safe enough, so the user abstains for sufficiently large $v$. When $\sigma_\varepsilon^2\le (1-\kappa)q$, the hardest tasks are handled by paying for reasoning rather than abstaining.

\begin{corollary}[Comparative statics]\label{cor:reasoning_cs}
Under Proposition~\ref{prop:reasoning_cutoff}:
\begin{enumerate}[label=(\roman*)]
\item $v_L$ is increasing in $\kappa$ (more expensive reasoning), in $q$ (lower stakes), and in $\sigma_\varepsilon^2$ (worse reasoning quality).
\item When $v_H<+\infty$, the feasibility cutoff $v_H$ is increasing in $q$ and decreasing in $\kappa$ and $\sigma_\varepsilon^2$.
\end{enumerate}
\end{corollary}

\subsection{Blind choice and the option value of selective reasoning}

A blind user cannot condition on $\sigma^2(x)$ and must choose an inference mode as a blanket policy, analogously to blind adoption in Section~\ref{sec:model}.
Let $\bar U_F(R)=1-1/R$ be the expected fast-mode payoff (Section~\ref{sec:complements}).
Under Definition~\ref{def:reasoning_noisy_anchor}, the expected payoff from always reasoning is
\begin{equation}\label{eq:UR_bar}
\bar U_T(R,\kappa,\sigma_\varepsilon^2)
\;\equiv\;
\E\!\left[U_T(x)\right]
=
1-\kappa-\frac{\E[\sigma_T^2(x)]}{q}.
\end{equation}
In the Brownian--Poisson baseline (Sections~\ref{sec:model}--\ref{sec:scaling}), it is convenient to write this in terms of the scale-free normalised variance $Z\equiv 3\lambda\sigma^2(x)$ defined in \eqref{eq:Z_def}, whose distribution is scale-invariant under Poisson coverage (Lemma~\ref{lem:scale_invariance}).
Let $\rho\equiv 3\lambda\sigma_\varepsilon^2$ be the reasoning noise floor in the same units.
Then \eqref{eq:reasoning_variance} implies
\begin{equation}\label{eq:ZR}
Z_T \;\equiv\;3\lambda\sigma_T^2(x)
=\frac{Z\rho}{Z+\rho},
\qquad\text{and hence}\qquad
\bar U_T(R,\kappa,\rho)
=1-\kappa-\frac{1}{R}\E\!\left[\frac{Z\rho}{Z+\rho}\right].
\end{equation}
A blind user chooses the best blanket rule,
\[
U^{\text{blind}}_{\text{modes}}(R,\kappa,\rho)
\;\equiv\;
\max\Big\{0,\ \bar U_F(R),\ \bar U_T(R,\kappa,\rho)\Big\}.
\]
By contrast, a calibrated user can choose the best mode task-by-task, earning
\[
U^{\text{cal}}_{\text{modes}}(R,\kappa,\rho)
\;\equiv\;
\E\!\left[\max\{0,\ U_F(x),\ U_T(x)\}\right].
\]
The gap $U^{\text{cal}}_{\text{modes}}-U^{\text{blind}}_{\text{modes}}$ is the option value of \emph{selective} reasoning.

\begin{proposition}[Jaggedness creates option value for reasoning]\label{prop:reasoning_option_value}
For any joint distribution of $(\sigma^2(x),\sigma_T^2(x))$,
\[
U^{\text{cal}}_{\text{modes}}\ \ge\ U^{\text{blind}}_{\text{modes}},
\]
with strict inequality whenever (i) $\Var(\sigma^2(x))>0$ and (ii) there is positive probability that different modes are
optimal at different realisations of $\sigma^2(x)$.
\end{proposition}

\begin{proof}
The function $g(a,b)\equiv \max\{0,a,b\}$ is convex in $(a,b)$ as the pointwise maximum of affine functions.
Therefore, Jensen's inequality gives
\[
\E[g(U_F(x),U_T(x))]\ \ge\ g(\E[U_F(x)],\E[U_T(x)]).
\]
The left-hand side is $U^{\text{cal}}_{\text{modes}}$ and the right-hand side equals
$\max\{0,\bar U_F,\bar U_T\}=U^{\text{blind}}_{\text{modes}}$.
Strictness holds whenever $g$ is not affine on the support of $(U_F,U_T)$, i.e.\ whenever the argmax switches across states.
\end{proof}

\noindent Proposition~\ref{prop:reasoning_option_value} is the ``reasoning'' analogue of the calibration value in Section~\ref{sec:calibration}. If tasks were homogeneous ($\sigma^2(x)$ constant), a blanket reasoning choice would be as good as task-level selection. Jaggedness creates dispersion in $\sigma^2(x)$, which makes per-task reasoning decisions valuable: the user can pay $\kappa$ only on the tasks that need it.

\subsection{Reasoning and the inspection paradox}

Section~\ref{sec:scaling} showed that user experience is length-biased toward long gaps (Proposition~\ref{prop:inspection_paradox}).
Because the direct benefit of reasoning is increasing in local uncertainty (Proposition~\ref{prop:reasoning_properties}), the inspection paradox systematically shifts the realised value of reasoning relative to evaluations that underweight long gaps.

To make the comparison sharp, consider two sampling schemes for gaps in the Poisson baseline.
A ``benchmark'' that samples a random \emph{gap} sees $X\sim \text{Exp}(\lambda)$, whereas a user sampling a random
\emph{location} sees a length-biased gap $X^\ast\sim \text{Gamma}(2,\lambda)$ (Proposition~\ref{prop:inspection_paradox}).
Conditional on a gap of length $X$ and a relative position $t\sim \text{Uniform}(0,1)$ within the gap,
baseline variance is $v=Xt(1-t)$ and reasoning variance is given by \eqref{eq:reasoning_variance}.

\begin{theorem}[Inspection-paradox amplification of reasoning value]\label{thm:reasoning_inspection}
Fix $(q,\kappa,\sigma_\varepsilon^2)$ and define the gross variance reduction
$\Delta(v)=v-\frac{v\sigma_\varepsilon^2}{v+\sigma_\varepsilon^2}$.
Let
\[
\bar \Delta^{\text{gap}}
\;\equiv\;
\E\!\left[\Delta\!\big(Xt(1-t)\big)\right],
\qquad
\bar \Delta^{\text{use}}
\;\equiv\;
\E\!\left[\Delta\!\big(X^\ast t(1-t)\big)\right],
\]
where in both expectations $t\sim\text{Uniform}(0,1)$ is independent of the gap length, and
$X\sim\text{Exp}(\lambda)$ while $X^\ast\sim\text{Gamma}(2,\lambda)$.
Then $\bar \Delta^{\text{use}}>\bar \Delta^{\text{gap}}$.
Equivalently, for any fixed reasoning cost $\kappa$, gap-uniform evaluations understate the expected net gain
$\E[V_T(x)]$ from reasoning in actual use.
\end{theorem}

\begin{proof}
For every fixed $t\in(0,1)$, the map $x\mapsto \Delta(x\,t(1-t))$ is strictly increasing in $x$.
Moreover, $X^\ast$ first-order stochastically dominates $X$ in the Poisson case:
$\Pr(X^\ast>z)=e^{-\lambda z}(1+\lambda z)>\Pr(X>z)=e^{-\lambda z}$ for all $z>0$.
Hence for each $t$,
$\E[\Delta(X^\ast t(1-t))\mid t]>\E[\Delta(X t(1-t))\mid t]$,
and integrating over $t$ yields the claim.
\end{proof}

\noindent  Theorem~\ref{thm:reasoning_inspection} is the reasoning analogue of the ``benchmark wedge'' in
Section~\ref{sec:complements}. Benchmarks underweight the longest gaps and therefore understate the \emph{expected} value of costly reasoning and, consequently, the welfare gains from improvements in reasoning quality. In the special case where the optimal reasoning region is an upper tail (i.e.\ $v_H=+\infty$ in Proposition~\ref{prop:reasoning_cutoff}), benchmarks also understate how frequently users optimally choose reasoning.

\subsection{Interactions with scale and regularity}

Reasoning introduces a design margin that is qualitatively different from the training-time scale. Scale ($\lambda$) reduces $\sigma^2(x)$ everywhere (Proposition~\ref{prop:scaling_law}), whereas reasoning reduces variance \emph{conditionally} and at a per-task cost. The two interact through the distribution of local variances: scaling changes how frequently the user encounters tasks with $v$ above the reasoning-entry cutoff $v_L$ (Proposition~\ref{prop:reasoning_cutoff}), and reasoning changes how costly the remaining long-gap tail is.

A simple way to see this is via the limiting behaviour of $\bar U_T$ in \eqref{eq:ZR}. Because $Z$ is scale-free under Poisson coverage, the dependence of always-reasoning performance on $\lambda$ is driven by $\rho=3\lambda\sigma_\varepsilon^2$.

\begin{proposition}[Always-reasoning limits under Poisson coverage]\label{prop:reasoning_scale_limits}
In the Brownian--Poisson baseline, if the user reasons on every task, then:
\begin{enumerate}[label=(\roman*)]
\item (Sparse scale) As $\lambda\to 0$ (equivalently $\rho\to 0$),
\[
\E[\sigma_T^2(x)]\to \sigma_\varepsilon^2
\qquad\text{and}\qquad
\bar U_T \to 1-\kappa-\frac{\sigma_\varepsilon^2}{q}.
\]
\item (Dense scale) As $\lambda\to\infty$ (equivalently $\rho\to\infty$),
\[
\E[\sigma_T^2(x)]\sim \frac{1}{3\lambda}
\qquad\text{and}\qquad
\bar U_T \sim 1-\kappa-\frac{1}{R}.
\]
\end{enumerate}
\end{proposition}

\begin{proof}
Using \eqref{eq:ZR},
\[
\E[\sigma_T^2(x)]
=\frac{1}{3\lambda}\E\!\left[\frac{Z\rho}{Z+\rho}\right]
=\sigma_\varepsilon^2\;\E\!\left[\frac{Z}{Z+\rho}\right].
\]
As $\rho\to 0$, $\frac{Z}{Z+\rho}\uparrow 1$ pointwise and is bounded by $1$, so dominated convergence gives
$\E\!\left[\frac{Z}{Z+\rho}\right]\to 1$ and hence $\E[\sigma_T^2(x)]\to\sigma_\varepsilon^2$.
As $\rho\to\infty$, $\frac{Z\rho}{Z+\rho}\to Z$ pointwise and is dominated by $Z$. By dominated convergence,
$\E\!\left[\frac{Z\rho}{Z+\rho}\right]\to \E[Z]=1$, and therefore
\[
\E[\sigma_T^2(x)]
=\frac{1}{3\lambda}\E\!\left[\frac{Z\rho}{Z+\rho}\right]
\sim \frac{1}{3\lambda}.
\]
Since $\rho=3\lambda\sigma_\varepsilon^2$, this is exactly the dense-scale limit as $\lambda\to\infty$ (holding
$\sigma_\varepsilon^2$ fixed).
Substituting into \eqref{eq:UR_bar} yields the stated limits.
\end{proof}

\noindent Proposition~\ref{prop:reasoning_scale_limits} shows that reasoning can act as a partial substitute for training-time scale in the low-$\lambda$ regime: it prevents the long-gap tail from exploding by imposing an error floor $\sigma_\varepsilon^2$. But reasoning does not ``beat'' scale in the high-$\lambda$ regime: once typical $\sigma^2(x)$ is well below the reasoning floor, reasoning does little besides impose the per-task cost $\kappa$.

Regularity interacts similarly. Without reasoning, the experienced mean variance scales with the second moment of gap lengths (Section~\ref{sec:complements}). With a reasoning floor, extremely long gaps contribute approximately $\sigma_\varepsilon^2$ per unit of task-space length rather than exploding quadratically, so the welfare return to targeting the longest gaps is attenuated on the \emph{intensive} margin. At the same time, if reasoning is costly and used selectively (Proposition~\ref{prop:reasoning_cutoff}), increasing regularity reduces how often users enter the long-gap tail where they must pay $\kappa$. This creates a natural three-way design space: training-time scale $(\lambda)$, coverage regularity (Section~\ref{sec:complements}), and inference-side reasoning quality $(\kappa,\sigma_\varepsilon^2)$.

In sum, reasoning does not eliminate jaggedness: it changes its economic geometry. It caps the losses from the longest gaps, but the inspection paradox ensures that those same gaps still dominate both (i) how often reasoning is invoked and (ii) how salient failures remain.

\section{Mastery: Learning Local Reliability}\label{sec:mastery}

Calibration (Section~\ref{sec:calibration}) is an information benchmark: the user observes $\sigma^2(x)$ and can
abstain on locally unreliable tasks. In most deployments, users do not receive such a clean reliability signal.
Instead, they invest in \emph{mastery}: learning, through repeated use, experimentation, and verification, which
task framings are safe, which are brittle, and how far successful patterns generalise.

This section formalises mastery as a learning problem and then studies its interaction with scaling investment.
The key message is not that mastery is ``better'' than scale, or vice versa. It is that scaling affects \emph{technology}
while mastery affects \emph{discoverability}: scaling can be real yet economically latent when improvements occur
in regions that users have learned to avoid or cannot infer about.

\subsection{Mastery as learning a reliability map}

In this section, the subscript $t=1,2,\ldots$ is a time index for sequential evaluations. This is unrelated to the within-gap coordinate $t\in[0,1]$ used earlier when averaging inside a gap.

Fix a model version (a fixed $\lambda$). Let $\mathcal{Z}$ denote the space of \emph{task representations} that a user
can induce: prompts, tool configurations, retrieval settings, input formats, and any other design degrees of freedom
that can move the realised interaction to a different location in task space.\footnote{This interpretation matters:
mastery is not merely ``knowing whether to trust the model'' at a fixed task. It is also learning how to \emph{move}
tasks into more reliable regions by rewriting, decomposing, adding context, or changing tools. Modelling this as a
choice of $x\in\mathcal{Z}$ is a reduced-form way to capture prompt libraries, workflow engineering, and standard
operating procedures.}

Let $V:\mathcal{Z}\to\R_+$ denote the model's local mean-squared error (or more generally, a local unreliability
index). In the Brownian--Poisson baseline, $V(x)=\sigma^2(x)$ is the local posterior variance induced by the hidden
knowledge-point configuration. The user does not observe $V(x)$ directly and must learn it.

From this point onward, $V(\cdot)$ is treated as an arbitrary (possibly jagged) reliability map. In the Brownian--Poisson benchmark, $V(x)=\sigma^2(x)$ is generated by hidden knowledge gaps; in applications, it could instead reflect any task-dependent error process. Assumption~\ref{ass:gp} is a model of the user's \emph{beliefs} about $V(\cdot)$ (and may be misspecified); the information-gain term $\gamma_t$ is a property of the kernel $k$ that governs how quickly subjective uncertainty shrinks under feedback.

\begin{assumption}[Feedback and verification noise]\label{ass:feedback}
When the user evaluates the model at $x_t\in\mathcal{Z}$ (e.g.\ by verifying against ground truth or by running an
audited test), they observe
\[
y_t = V(x_t) + \eta_t,
\]
where $(\eta_t)_{t\ge 1}$ are i.i.d.\ conditional on $(x_s)_{s\ge 1}$ and $V$, with $\eta_t\sim\mathcal{N}(0,\sigma_n^2)$.
\end{assumption}

\noindent  Even in settings with ``objective'' evaluation, reliability is observed with noise: tasks vary within the same prompt template, evaluation is imperfect, and organisations use proxy checks (spot audits, consensus checks, or downstream error signals) rather than full ground truth. We assume Gaussian noise for analytical convenience and because the mutual-information identity used in Theorem~\ref{thm:learning_rate} is exact under a Gaussian likelihood.

\begin{assumption}[Beliefs: a Gaussian-process prior]\label{ass:gp}
The user models $V(\cdot)$ as a sample from a Gaussian process prior $V\sim \GP(0,k)$ with bounded kernel
$k(x,x)\le 1$.
\end{assumption}

\noindent A Gaussian process (or GP) prior is a tractable representation of a sophisticated user's belief that ``similar tasks have similar reliability. The kernel $k$ encodes what the user regards as similarity: changing a few words in a prompt may be viewed as a small perturbation (high $k$), while switching domains or tools may be treated as essentially unrelated (low $k$). This is not merely statistical convenience: interface design (templates, tool affordances, workflow structure) and organisational practice (prompt libraries, input schemas) shape $k$ by determining which transformations users regard as meaningful and transferable.

Let $\mu_t(x)\equiv \E[V(x)\mid \mathcal{F}_t]$ and $s_t^2(x)\equiv \Var(V(x)\mid \mathcal{F}_t)$ denote the GP posterior mean and variance after $t$ feedback observations, where $\mathcal{F}_t$ is the sigma-algebra generated by the history. To avoid confusion with the model's local error $\sigma^2(x)$, we reserve $\sigma^2(x)$ for the \emph{true} local error map and use $s_t^2(x)$ for the user's \emph{uncertainty} about that map.

A natural measure of mastery is worst-case uncertainty:
\begin{equation}\label{eq:mastery_metric}
m_t \;\equiv\; \sup_{x\in\mathcal{Z}} s_t^2(x).
\end{equation}
A large $m_t$ means there remain regions in which the user cannot reliably distinguish safe from unsafe delegation.

\subsection{An optimistic benchmark: how fast could mastery learn?}

If the user is willing to run deliberate experiments, a simple strategy is \emph{uncertainty sampling}:
\begin{equation}\label{eq:uncertainty_sampling}
x_{t+1}\in \arg\max_{x\in\mathcal{Z}} s_t^2(x).
\end{equation}
This formalises systematic stress-testing: probe where you are least sure. It is an optimistic benchmark because it
assumes the user can afford to sample points that are potentially unsafe; in high-stakes settings, that assumption will
fail (and that failure is precisely where the interaction with scaling becomes interesting).

Define the \emph{maximum information gain} after $t$ samples,
\begin{equation}\label{eq:gamma_def}
\gamma_t \;\equiv\; \max_{A\subset \mathcal{Z}: |A|=t} I\!\left(V; y_A\right),
\end{equation}
where $y_A$ denotes the vector of noisy observations at design points in $A$.

\begin{theorem}[A speed limit on mastery]\label{thm:learning_rate}
Under Assumptions~\ref{ass:feedback}--\ref{ass:gp}, suppose the user chooses $x_{t+1}$ by uncertainty sampling
\eqref{eq:uncertainty_sampling}. Define the hitting time
\[
t_0 \;\equiv\; \inf\{t\in\mathbb{N}_0 : m_t\le \sigma_n^2\},
\]
with the convention $t_0=\infty$ if the set is empty. Then for all integers
$t\ge t_0+1$ (when $t_0<\infty$),
\[
m_t \le \frac{4\sigma_n^2\,\gamma_t}{t-t_0}.
\]
In particular, asymptotically $m_t = O(\sigma_n^2 \gamma_t/t)$.
\end{theorem}

\begin{proof}
For GP regression with Gaussian observation noise variance $\sigma_n^2$, the mutual information between $V$ and
observations at points $x_1,\dots,x_t$ satisfies
\[
I(V; y_{1:t})=\frac{1}{2}\sum_{s=1}^t \log\!\left(1+\frac{s_{s-1}^2(x_s)}{\sigma_n^2}\right)
\]
(e.g.\ \citet{rasmussen2006gaussian}). Since $\gamma_t$ is the maximum over all designs, $I(V;y_{1:t})\le \gamma_t$.

For $a\in[0,1]$, $\log(1+a)\ge a/2$. By definition of $t_0$ and monotonicity of posterior variances, for all $s\ge t_0+1$
we have $s_{s-1}^2(x_s)\le m_{s-1}\le m_{t_0}\le \sigma_n^2$, hence $s_{s-1}^2(x_s)/\sigma_n^2\in[0,1]$ and
\[
\log\!\left(1+\frac{s_{s-1}^2(x_s)}{\sigma_n^2}\right)\ge \frac{1}{2}\frac{s_{s-1}^2(x_s)}{\sigma_n^2}.
\]
Therefore,
\[
\gamma_t \ge I(V;y_{1:t})
\ge \frac{1}{2}\sum_{s=t_0+1}^t \log\!\left(1+\frac{s_{s-1}^2(x_s)}{\sigma_n^2}\right)
\ge \frac{1}{4\sigma_n^2}\sum_{s=t_0+1}^t s_{s-1}^2(x_s).
\]
Under uncertainty sampling, $s_{s-1}^2(x_s)=m_{s-1}$ and $\{m_s\}$ is non-increasing, hence
\[
m_t \le \frac{1}{t-t_0}\sum_{s=t_0+1}^t m_{s-1}
= \frac{1}{t-t_0}\sum_{s=t_0+1}^t s_{s-1}^2(x_s)
\le \frac{4\sigma_n^2\,\gamma_t}{t-t_0}.
\]
\end{proof}

\noindent 
Theorem~\ref{thm:learning_rate} gives an \emph{upper bound} on residual uncertainty under uncertainty sampling:
worst-case posterior uncertainty about local reliability satisfies $m_t=O(\sigma_n^2\gamma_t/t)$ (up to constants and the noise floor).
The term $\gamma_t$ is the intrinsic complexity of the reliability map under the user's similarity notion. If users treat the task space as high-dimensional or weakly structured, $\gamma_t$ grows quickly and mastery is slow.
If the task space is effectively low-dimensional (e.g.\ a narrow workflow with stable templates), $\gamma_t$ grows slowly, and mastery can be fast.

\begin{proposition}[Information-gain growth rates]\label{prop:mig_rates}
Let $\mathcal{Z}\subset\R^d$ be compact and suppose $k$ is either:
\begin{enumerate}[label=(\alph*)]
\item a squared-exponential (RBF) kernel, or
\item a Mat\'ern kernel with smoothness parameter $\nu>0$.
\end{enumerate}
Then the maximum information gain satisfies:
\begin{align*}
\text{(RBF)}\qquad &\gamma_t = O\!\left((\log t)^{d+1}\right) \qquad \text{\citep{srinivas2010gaussian}},\\
\text{(Mat\'ern)}\qquad &\gamma_t = \tilde O\!\left(t^{\frac{d}{2\nu+d}}\right) \qquad \text{\citep{vakili2021information}},
\end{align*}
where $\tilde O(\cdot)$ suppresses polylogarithmic factors.
\end{proposition}

\noindent 
Combining Proposition~\ref{prop:mig_rates} with Theorem~\ref{thm:learning_rate} yields a sample-complexity message: to drive worst-case uncertainty below $\varepsilon$, one needs on the order of $t\gtrsim \sigma_n^2 \gamma_t/\varepsilon$ evaluations. This aligns with observed practice: users can quickly learn reliable prompt playbooks for structured domains (low effective dimension), but struggle to develop stable policies for open-ended tasks (high effective dimension).

\subsection{Mastery as delegation under uncertainty}

Mastery matters economically because it determines \emph{where} the user delegates. Given local error $V(x)$,
per-task expected utility from using the AI is $U(x)=1-V(x)/q$ (outside option normalised to $0$).
If the user knew $V(x)$ exactly, the optimal rule would match perfect calibration: delegate if and only if $V(x)\le q$.
This is a useful upper bound.

\begin{lemma}[Perfect mastery replicates calibration]\label{lem:perfect_mastery}
If the user knows $V(x)$ pointwise, then the optimal delegation rule is $\pi^\star(x)=\1\{V(x)\le q\}$ and expected
utility is
\[
U_M^\star(\lambda,q)=\E\!\left[\left(1-\frac{V(x)}{q}\right)_+\right].
\]
In the baseline where $V(x)=\sigma^2(x)$, this coincides with calibrated utility $U_C(\lambda,q)$ from
Section~\ref{sec:calibration}.
\end{lemma}

\begin{proof}
Pointwise, the user chooses between $U(x)$ and $0$, so the optimal payoff is $\max\{U(x),0\}=(1-V(x)/q)_+$.
Taking expectations yields the expression. When $V(x)=\sigma^2(x)$, this is exactly the definition of $U_C$.
\end{proof}

\noindent In practice, users face liability, verification costs, and asymmetric loss: a single bad failure can dominate many small gains. A reduced-form way to capture this is a \emph{conservative delegation rule} based on an upper confidence bound for $V(x)$.

Fix $\beta>0$ and define the conservative estimate
\begin{equation}\label{eq:Vhat}
\hat V_t(x)\;\equiv\;\mu_t(x)+\sqrt{\beta}\,s_t(x).
\end{equation}
For suitable $\beta$ (chosen to deliver a desired confidence level), $\hat V_t(x)$ is a high-probability upper bound on
$V(x)$ under the GP model (the standard ``GP-UCB'' form; see \citet{srinivas2010gaussian}).

The induced \emph{perceived safe set} is
\begin{equation}\label{eq:safe_set}
S_t\;\equiv\;\{x\in \mathcal{Z}:\hat V_t(x)\le q\},
\end{equation}
and the corresponding conservative delegation policy is $\pi_t(x)=\1\{x\in S_t\}$.

The purpose of \eqref{eq:Vhat}--\eqref{eq:safe_set} is not to import the full bandit apparatus, but to encode an empirically important behavioural regularity: in high-stakes domains, organisations do not delegate based on mean beliefs alone. They require evidence that delegation is safe with high confidence, which shrinks the delegated region relative to the calibrated benchmark and makes learning \emph{endogenous} to past delegation.

\subsection{Scaling and the inertia of mastery}

We now examine how scaling investment interacts with mastery. Let $V_0$ denote the pre-update reliability map of a model with intensity $\lambda_0$, and let $V_1$ denote the post-update map after scaling to $\lambda_1>\lambda_0$. The user has learned under $V_0$ and carries those beliefs into the new version.

\begin{assumption}[Monotonic scaling]\label{ass:monotonic}
Scaling from $\lambda_0$ to $\lambda_1>\lambda_0$ adds coverage rather than reshuffling it. Consequently, local error weakly decreases pointwise:
\[
V_1(x)\le V_0(x)\qquad \text{for all }x\in\mathcal{Z}.
\]
\end{assumption}

\noindent 
This assumption isolates the \emph{informational} friction in mastery. If updates could arbitrarily worsen some regions, then slow adoption could be optimal for risk reasons even with perfect information. Assumption~\ref{ass:monotonic} rules that out: it ensures that previously safe regions remain safe, and that any failure to harvest improvements arises because the user does not \emph{learn} about newly improved regions. It is also consistent with the Brownian--Poisson baseline: adding knowledge points refines gaps and (weakly) reduces the Brownian-bridge variance everywhere.

Next, to sharpen the idea that learning generalises locally, we impose a finite inference radius.

\begin{assumption}[Local similarity / finite inference radius]\label{ass:local_similarity}
There exists $\ell>0$ such that the kernel satisfies $k(x,x')=0$ whenever $\|x-x'\|>\ell$.
\end{assumption}

\noindent 
Assumption~\ref{ass:local_similarity} rules out degenerate cases in which one observation pins down reliability
everywhere. It can be motivated either by bounded transfer across prompt variants (users treat sufficiently different task framings as unrelated), or by the local nature of interpolation risk in the baseline (uncertainty at a point is governed by nearby anchors rather than distant ones). Kernels without compact support (e.g.\ RBF or Mat\'ern) generate the same qualitative logic: spillovers decay with distance, so the ``hidden scale'' result below becomes an arbitrarily accurate approximation as improvements occur farther from the user's visited set.

Let $\mathcal{A}_t=\{x_1,\dots,x_t\}$ be the set of evaluation locations up to time $t$, and define its $\ell$-neighbourhood
\begin{equation}\label{eq:neighbourhood}
N_\ell(\mathcal{A}_t)\;\equiv\;\{x\in\mathcal{Z}:\min_{a\in\mathcal{A}_t}\|x-a\|\le \ell\}.
\end{equation}
Under Assumption~\ref{ass:local_similarity}, points outside $N_\ell(\mathcal{A}_t)$ are statistically disconnected from the user's experience at time $t$: absent deliberate exploration, their posterior remains at its prior.

The next proposition formalises the ``abstention trap'': improvements can be technologically real yet economically invisible.

\begin{proposition}[Zero marginal returns to hidden scale]\label{prop:abstention_trap}
Fix a belief state $\mathcal{F}_t$ with data locations $\mathcal{A}_t$ and safe set $S_t$.
Consider two post-update error maps $V_1$ and $\tilde V_1$ satisfying
$\tilde V_1(x)\le V_1(x)$ for all $x\in\mathcal{Z}$ and
\begin{equation}\label{eq:hidden_condition}
\tilde V_1(x)=V_1(x)\quad \text{for all }x\in N_\ell(\mathcal{A}_t).
\end{equation}
Suppose the user carries over $\mathcal{F}_t$ as their prior at the update and, after the update, observes feedback
only at locations they delegate (i.e.\ only at points in $S_s$ for $s\ge t$).
Define the stopping time
\[
\tau\;\equiv\;\inf\{s>t:\ x_s\notin N_\ell(\mathcal{A}_t)\},
\]
with the convention $\inf\emptyset=\infty$.
Then, under Assumptions~\ref{ass:monotonic} and \ref{ass:local_similarity}, the induced sequences of posteriors and
delegation policies are identical under $V_1$ and $\tilde V_1$ on the event $\{\tau=\infty\}$, and more generally are
identical \emph{up to time $\tau$} on every sample path.
In particular, any improvement in $\tilde V_1$ relative to $V_1$ that is supported outside $N_\ell(\mathcal{A}_t)$ has
zero marginal effect on outcomes until the user first samples outside that neighbourhood, and has zero marginal effect
forever if $\tau=\infty$.
\end{proposition}

\begin{proof}
Condition \eqref{eq:hidden_condition} implies that the conditional law of feedback $y=V(x)+\eta$ is identical under
$V_1$ and $\tilde V_1$ at every $x\in N_\ell(\mathcal{A}_t)$.
By Assumption~\ref{ass:local_similarity} and the definition of $\tau$, for all $s<\tau$ the user only delegates within
$N_\ell(\mathcal{A}_t)$, so the realised dataset up to time $s$ contains only observations from that neighbourhood.
Because $V_1$ and $\tilde V_1$ coincide on $N_\ell(\mathcal{A}_t)$, the realised feedback sequence (and hence the GP
posterior restricted to $N_\ell(\mathcal{A}_t)$) is identical under the two maps for all $s<\tau$.

We prove by induction on $s$ that, on any sample path for which $x_{t+1},\dots,x_s\in N_\ell(\mathcal{A}_t)$, the realised post-update datasets up to time $s$ are identical under $V_1$ and $\tilde V_1$, and hence the posteriors
$(\mu_s,s_s)$ (restricted to $N_\ell(\mathcal{A}_t)$) are identical under the two maps. The base case $s=t$ holds by the shared prior $\mathcal{F}_t$.
For the induction step, suppose the claim holds at time $s$ and that $x_{s+1}\in N_\ell(\mathcal{A}_t)$.
Because the user's delegation and evaluation choices are measurable functions of their current posterior, the induction hypothesis implies the same choice $x_{s+1}$ under both maps.
Since $V_1$ and $\tilde V_1$ coincide at $x_{s+1}$, the conditional law of $y_{s+1}$ is the same under both maps, so the realised dataset remains identical and so does the posterior at time $s+1$.

Thus, the posterior and hence the delegation policy coincide at all times $s<\tau$.
If $\tau=\infty$ the processes coincide for all times.
Since realised utility depends only on the delegated set and on $V(\cdot)$ on that set, improvements outside
$N_\ell(\mathcal{A}_t)$ have no marginal effect until $\tau$ (and never when $\tau=\infty$).
\end{proof}

\noindent 
Mastery data are generated where the user actually operates. Under conservative delegation, ``unsafe'' regions are not sampled, so beliefs there do not update. Scaling can, therefore, create \emph{hidden capabilities}: tasks that have become safe in truth but remain excluded by the user's inherited mastery map. Proposition~\ref{prop:abstention_trap} shows that, under local generalisation, improvements that lie outside the user's informational reach have zero realised marginal return until the user deliberately ventures outward. If the user never does, those returns are zero forever.

\subsection{Frontier spillovers}

Proposition~\ref{prop:abstention_trap} does not imply that scaling is generically useless.
Users often operate near a frontier where $\hat V_t(x)$ is close to $q$. Improvements \emph{within} the region they do visit can shift beliefs about nearby tasks and expand the safe set. The next lemma records the GP update and makes locality explicit.

\begin{lemma}[Local belief updating under GP regression]\label{lem:gp_update}
Under Assumptions~\ref{ass:feedback}--\ref{ass:gp}, after observing $y_{t+1}=V(x_{t+1})+\eta_{t+1}$ at location $x_{t+1}$, the posterior mean and variance satisfy
\begin{align}
\mu_{t+1}(x) &= \mu_t(x) + \frac{k_t(x,x_{t+1})}{k_t(x_{t+1},x_{t+1})+\sigma_n^2}\Big(y_{t+1}-\mu_t(x_{t+1})\Big),
\label{eq:gp_mean_update}\\
s_{t+1}^2(x) &= s_t^2(x) - \frac{k_t(x,x_{t+1})^2}{k_t(x_{t+1},x_{t+1})+\sigma_n^2},
\label{eq:gp_var_update}
\end{align}
where $k_t(\cdot,\cdot)$ is the posterior covariance kernel at time $t$.
In particular, if $k_t(x,x_{t+1})=0$ then $\mu_{t+1}(x)=\mu_t(x)$ and $s_{t+1}(x)=s_t(x)$.
\end{lemma}

\begin{proof}
These are the standard GP regression update formulas (e.g.\ \citet[Ch.\ 2]{rasmussen2006gaussian}). The final claim follows immediately from \eqref{eq:gp_mean_update}--\eqref{eq:gp_var_update}.
\end{proof}

\noindent We now formalise the ``frontier spillover'' mechanism: observing improvement at a queried point reduces pessimism and tightens uncertainty nearby, potentially expanding the perceived safe set.

\begin{proposition}[Frontier spillovers of scaling]\label{prop:frontier_spillovers}
Suppose after a scaling improvement, the true error map shifts from $V_0$ to $V_1$ with $V_1\le V_0$ pointwise.
Consider a point $x^\dagger\in \mathcal{S}_t$ that the user evaluates immediately after the update, and define the
(local) \emph{surprise improvement relative to the user's pre-update belief}
\[
\Delta_t^\dagger \;\equiv\; \mu_t(x^\dagger)-V_1(x^\dagger).
\]
Then, conditional on $\mathcal{F}_t$ and $x_{t+1}=x^\dagger$,
\begin{equation}\label{eq:frontier_bound}
\E[\hat V_{t+1}(x)\mid \mathcal{F}_t, x_{t+1}=x^\dagger]
\le \hat V_t(x)- \frac{k_t(x,x^\dagger)}{k_t(x^\dagger,x^\dagger)+\sigma_n^2}\,\Delta_t^\dagger.
\end{equation}
\end{proposition}

\begin{proof}
Conditional on $\mathcal{F}_t$ and $x_{t+1}=x^\dagger$, the only remaining randomness in $y_{t+1}$ is $\eta_{t+1}$,
so $\E[y_{t+1}\mid \mathcal{F}_t, x_{t+1}=x^\dagger]=V_1(x^\dagger)$. Hence the expected innovation is
\[
\E[y_{t+1}-\mu_t(x^\dagger)\mid \mathcal{F}_t, x_{t+1}=x^\dagger]
= V_1(x^\dagger)-\mu_t(x^\dagger)=-\Delta_t^\dagger.
\]
Plugging into Lemma~\ref{lem:gp_update} yields
$\E[\mu_{t+1}(x)\mid \mathcal{F}_t, x_{t+1}=x^\dagger]=\mu_t(x)-\frac{k_t(x,x^\dagger)}{k_t(x^\dagger,x^\dagger)+\sigma_n^2}\Delta_t^\dagger$.
Since $s_{t+1}(x)\le s_t(x)$, we have $\hat V_{t+1}(x)\le \mu_{t+1}(x)+\sqrt{\beta}\,s_t(x)$, giving \eqref{eq:frontier_bound}.
\end{proof}

\noindent
The hidden scale is spatial: improvements far from where users operate do not move beliefs, but improvements on (or near) frontier tasks are informative about adjacent tasks and can expand the safe set. The strength and radius of this effect are governed by the kernel (what users regard as ``similar'') and by the noise level $\sigma_n^2$ (how costly or noisy verification is).

\subsection{The baseline geometry: safe fringes in Brownian gaps}

In the Brownian--Poisson baseline, the locality above has a simple geometric interpretation.
Between two adjacent knowledge points separated by a gap of length $X$, the local variance is the Brownian-bridge
expression \eqref{eq:bridge-variance}: $\sigma^2(x)=Xt(1-t)$ where $t\in[0,1]$ is the relative position within the gap.
The set of locally nonnegative-value tasks in that gap depends only on $X$.

\begin{lemma}[Safe fringe within a Brownian-bridge gap]\label{lem:bridge_fringe}
Fix a gap of length $X$ between adjacent knowledge points in the baseline model.
The set of locations in the gap where using the AI yields nonnegative expected utility,
$\{x:\sigma^2(x)\le q\}$, is the union of two intervals adjacent to the gap endpoints.
If $X\le 4q$ the entire gap is safe.
If $X>4q$, the safe intervals have length
\begin{equation}\label{eq:fringe_length}
d(X,q)\;\equiv\;\frac{1}{2}\Big(X-\sqrt{X(X-4q)}\Big),
\end{equation}
so the unsafe region is the middle interval of length $X-2d(X,q)$.
Moreover, $d(X,q)$ is strictly decreasing in $X$ for $X>4q$.
\end{lemma}

\begin{proof}
Index the gap as $[0,X]$. At distance $s\in[0,X]$ from the left endpoint, $\sigma^2(s)=s(X-s)/X$.
The condition $\sigma^2(s)\le q$ is equivalent to $s(X-s)\le qX$, i.e.\ $s^2-Xs+qX\ge 0$.
If $X\le 4q$, the discriminant is nonpositive and the inequality holds for all $s\in[0,X]$.
If $X>4q$, the quadratic has roots
$s_\pm=\frac{1}{2}\big(X\pm\sqrt{X(X-4q)}\big)$,
and $\sigma^2(s)\le q$ holds for $s\in[0,s_-]\cup[s_+,X]$.
Thus the safe fringe length is $d(X,q)=s_-$. Differentiating \eqref{eq:fringe_length} yields $d_X(X,q)<0$ for $X>4q$.
\end{proof}

\noindent 
Scaling works mechanically by shrinking gaps. In this baseline, shrinking a gap expands the safe fringe deterministically. This is the simplest case of frontier spillovers: even if a user only operates near existing anchors (safe regions), adding anchors nearby expands the region where delegation is locally nonnegative.

\subsection{Implications for scaling investment}

Mastery makes the realised return to scaling \emph{state-dependent} because scaling changes the underlying map $V$ while mastery determines which parts of that map are actually harvested.

Let $x$ denote a random task drawn from the arrival distribution on $\mathcal{Z}$.
Given a delegation policy $\pi(\cdot)\in\{0,1\}^{\mathcal{Z}}$, realised expected utility under error map $V$ is
\[
\bar U(V,\pi)\;\equiv\;\E\!\left[\left(1-\frac{V(x)}{q}\right)\pi(x)\right].
\]
For any pre-update policy $\pi_t$ and post-update policy $\pi_{t+1}$,
\begin{align}
\bar U(V_1,\pi_{t+1})-\bar U(V_0,\pi_t)
&=
\E\!\left[\frac{V_0(x)-V_1(x)}{q}\,\pi_t(x)\right]\label{eq:two_margins_mastery}\\
&\qquad\qquad + \E\!\left[\left(1-\frac{V_1(x)}{q}\right)\big(\pi_{t+1}(x)-\pi_t(x)\big)\right].\nonumber
\end{align}
The first term is the \emph{intensive margin}: improvements harvested on tasks the user was already delegating.
The second term is the \emph{extensive margin}: additional value from expanding the delegated set as beliefs update.

Proposition~\ref{prop:abstention_trap} is a sharp constraint on the extensive margin: improvements outside the user's inferential neighbourhood cannot change $\pi_{t+1}$ unless the user explores outward. Proposition~\ref{prop:frontier_spillovers} and Lemma~\ref{lem:bridge_fringe} show the opposite channel: improvements near the user's operating frontier can expand $\pi_{t+1}$ even without risky exploration.

Relative to Section~\ref{sec:complements}, which studies static returns to scale under fixed information, mastery adds a dynamic wedge:
\begin{itemize}[leftmargin=1.25em]
\item \textbf{Discoverability matters.} The realised marginal return to increasing $\lambda$ depends on whether improved
regions are visited or inferable from visited regions. Scaling that mainly improves ``elsewhere'' can have a low realised return in the short run, even if it improves true capability everywhere.
\item \textbf{Frontier targeting is disproportionately valuable.} Improvements concentrated near tasks users already touch,
and especially near tasks close to the safety boundary, can trigger safe-set expansion and unlock an extensive margin.
\item \textbf{Release policy interacts with mastery.} When re-verification and workflow adjustment are costly, users may
carry over conservative mastery maps across updates. Under monotonic scaling, this is safe, but it can strand gains.
This creates a rationale for pairing scaling with either (i) salient evidence of improvement that shifts beliefs, or
(ii) interface-level calibration tools that substitute for slow mastery.
\end{itemize}
In sum, monotonic scaling ensures that existing mastery is not destroyed, but it does not ensure that new capability is discovered. Mastery is, therefore, a complement to scaling, not because it makes the model better, but because it governs how quickly and how broadly users can \emph{harvest} the improvements that scaling creates.

\section{Extensions and Applications}\label{sec:extensions}

The preceding sections develop the core economic logic of Artificial Jagged Intelligence (AJI): user-experienced quality
is governed by \emph{local reliability} and is shaped by the inspection paradox, so benchmark averages can be poor
predictors of lived value. This section develops two applications that use this logic in settings where the relevant
objects are not just ``a representative adopter'' but heterogeneous users and institutions.

The first application shows how AJI naturally generates \emph{vertical differentiation by reliability profile}: users with different stakes can rationally prefer different systems even when ``average capability'' rankings are clear.
The second application studies \emph{organisational adoption versus worker experience}: principals adopt using benchmark
signals while agents experience length-biased tasks, creating over-adoption and predictable trust dynamics.

\subsection{Vertical Differentiation by Reliability Profile}\label{sec:vertical_diff}

Users differ sharply in stakes. When errors are cheap to correct (high $q$), users care mainly about average capability.
When errors are costly or irreversible (low $q$), users care disproportionately about avoiding rare but extreme failures.
AJI makes this distinction economically meaningful because jagged landscapes generate a high-error tail that is
\emph{overweighted} in actual use (Proposition~\ref{prop:inspection_paradox}). As a result, there is a second dimension of quality beyond benchmark averages: a system's \emph{reliability profile}---how error is distributed across tasks, not just its mean.

In the baseline payoff $U(x)=1-\sigma^2(x)/q$, a blind user's expected value depends only on $\E[\sigma^2]$.
That is appropriate when the loss from error is approximately linear, and the user is effectively risk-neutral in error.
High-stakes deployment is rarely like that: liability, downstream cascades, and compliance constraints introduce
convex losses, making \emph{dispersion} in $\sigma^2(x)$ (tail risk) a first-order object. The extension below isolates this channel in a minimal way that preserves the model's structure.

\subsubsection{A minimal model of ``catastrophic'' stakes}

\begin{assumption}[Convex error losses]\label{ass:convex_losses}
When a user relies on provider $j$ at task $x$, their net payoff is
\begin{equation}\label{eq:convex_payoff}
u_j(x;q)
\;=\;
1-\frac{\sigma_j^2(x)}{q}
-\phi\left(\frac{\sigma_j^2(x)}{q}\right)^{\!2},
\qquad \phi\ge 0,
\end{equation}
with outside option normalised to $0$.
\end{assumption}

\noindent 
The quadratic term is a reduced-form way to capture the idea that large errors have disproportionately high costs in
high-stakes settings (legal exposure, reputational harm, safety incidents, cascading rework).
It is also the second-order term of a generic convex loss expansion around small errors.
The baseline model is nested at $\phi=0$, so all baseline results remain interpretable as the risk-neutral benchmark.

Under Assumption~\ref{ass:convex_losses}, a blind user's expected value from provider $j$ (before truncation at the outside
option) is
\begin{equation}\label{eq:Ubar_convex}
\bar U_j(q)
\;\equiv\;
\E[u_j(x;q)]
=
1-\frac{\E[\sigma_j^2(x)]}{q}-\phi\frac{\E[\sigma_j^4(x)]}{q^2}.
\end{equation}
The new object is the \emph{fourth moment} $\E[\sigma_j^4(x)]$, which is precisely where the reliability profile matters.
The next lemma computes these moments in the Brownian bridge environment and makes the role of the inspection paradox
transparent.

\begin{lemma}[Moments of experienced error]\label{lem:moments_sigma}
Fix a provider $j$ with gap-length distribution $X_j$ having mean $\E[X_j]=1/\lambda_j$.
Let $X_j^*$ denote the length-biased gap faced by a randomly arriving task (Proposition~\ref{prop:inspection_paradox}) and
let $t\sim\mathrm{Uniform}(0,1)$ be relative position in the gap, so that
$\sigma_j^2(x)=X_j^*\,t(1-t)$.
Then
\begin{align}
\E[\sigma_j^2(x)] &= \frac{\E[X_j^*]}{6}=\frac{\E[X_j^2]}{6\,\E[X_j]},\label{eq:m1_sigma}\\
\E[\sigma_j^4(x)] &= \frac{\E[(X_j^*)^2]}{30}=\frac{\E[X_j^3]}{30\,\E[X_j]}.\label{eq:m2_sigma}
\end{align}
\end{lemma}

\begin{proof}
Conditional on $X_j^*$, $\sigma_j^2=X_j^*\,t(1-t)$ with $t\sim \mathrm{Uniform}(0,1)$, so
$\E[t(1-t)]=\int_0^1 t(1-t)\,dt=1/6$ and
$\E[t^2(1-t)^2]=\int_0^1 t^2(1-t)^2\,dt=1/30$.
Taking expectations over $X_j^*$ yields
$\E[\sigma_j^2]=\E[X_j^*]/6$ and $\E[\sigma_j^4]=\E[(X_j^*)^2]/30$.
Finally, for any gap distribution $X$, the length-biased moments satisfy
$\E[(X^*)^m]=\E[X^{m+1}]/\E[X]$, giving \eqref{eq:m1_sigma}--\eqref{eq:m2_sigma}.
\end{proof}

\noindent Lemma~\ref{lem:moments_sigma} shows why the inspection paradox matters for ``catastrophic'' stakes.
Benchmarks that effectively weight gaps uniformly (one item per gap/domain) replace $X^*$ by $X$.
Under that gap-uniform sampling,
\[
\E_{\text{gap-unif}}[\sigma^2]=\frac{\E[X]}{6}=\frac{1}{6\lambda},
\qquad
\E_{\text{gap-unif}}[\sigma^4]=\frac{\E[X^2]}{30}.
\]
By contrast, user experience replaces $\E[X]$ by $\E[X^*]=\E[X^2]/\E[X]$ and $\E[X^2]$ by $\E[(X^*)^2]=\E[X^3]/\E[X]$.
Thus, AJI amplifies not only the mean error but also the higher moments that govern tail-sensitive losses.

\subsubsection{A tractable regularity family}

To connect the reliability profile to the regularity margin in Section~\ref{sec:complements}, it is convenient to impose a parametric family in which a single parameter controls dispersion.

\begin{assumption}[Gamma regularity family]\label{ass:gamma_gaps}
For provider $j$, gap lengths satisfy $X_j\sim \mathrm{Gamma}(a_j,\,a_j\lambda_j)$ (shape $a_j>0$, rate $a_j\lambda_j$), so that $\E[X_j]=1/\lambda_j$ and $\CV(X_j)=1/\sqrt{a_j}$.
\end{assumption}

\noindent 
The gamma family is a standard reduced form for spacings: it nests the Poisson benchmark at $a=1$ (exponential gaps),
and it converges to deterministic spacing as $a\to\infty$ (perfect regularity). It therefore provides a microfoundation for the regularity parameter in Section~\ref{sec:complements} while preserving closed-form moments.

Under Assumption~\ref{ass:gamma_gaps}, the moment expressions in Lemma~\ref{lem:moments_sigma} can be written directly in
terms of $(\lambda,\CV)$.

\begin{lemma}[Closed-form moments under gamma regularity]\label{lem:moments_gamma}
Under Assumption~\ref{ass:gamma_gaps}, writing $\CV_j\equiv \CV(X_j)$,
\begin{align}
\E[\sigma_j^2(x)] &= \frac{1+\CV_j^{2}}{6\lambda_j},\label{eq:Esig2_gamma}\\
\E[\sigma_j^4(x)] &= \frac{1+3\CV_j^{2}+2\CV_j^{4}}{30\,\lambda_j^{2}}
=\frac{(1+\CV_j^{2})(1+2\CV_j^{2})}{30\,\lambda_j^{2}}.\label{eq:Esig4_gamma}
\end{align}
Moreover, the inspection-paradox amplification of the fourth moment is
\[
\frac{\E[\sigma_j^4]_{\text{user}}}{\E[\sigma_j^4]_{\text{gap-unif}}}
=
1+2\CV_j^{2},
\]
so under Poisson gaps ($\CV_j=1$) user-experienced $\E[\sigma^4]$ is three times the gap-uniform benchmark value.
\end{lemma}

\begin{proof}
For $X\sim\mathrm{Gamma}(a,a\lambda)$, $\E[X^2]=(a+1)/(a\lambda^2)=(1+1/a)/\lambda^2=(1+\CV^2)/\lambda^2$,
so \eqref{eq:Esig2_gamma} follows from \eqref{eq:m1_sigma}.
Similarly $\E[X^3]=(a+1)(a+2)/(a^2\lambda^3)$, so
$\E[(X^*)^2]=\E[X^3]/\E[X]=(a+1)(a+2)/(a^2\lambda^2)$ and \eqref{eq:Esig4_gamma} follows from \eqref{eq:m2_sigma}.
For the amplification ratio, the user-experienced fourth moment is
\[
\E[\sigma^4]_{\text{user}}=\frac{1}{30}\E[(X^*)^2]=\frac{(1+\CV^2)(1+2\CV^2)}{30\lambda^2},
\]
while the gap-uniform fourth moment is
$\E[\sigma^4]_{\text{gap-unif}}=\frac{1}{30}\E[X^2]=\frac{1+\CV^2}{30\lambda^2}$.
Dividing gives $\E[\sigma^4]_{\text{user}}/\E[\sigma^4]_{\text{gap-unif}}=1+2\CV^2$.
\end{proof}

\subsubsection{Sorting by stakes}

Consider two providers $H$ and $L$ with reliability profiles induced by $(\lambda_H,\CV_H)$ and $(\lambda_L,\CV_L)$.
Think of $H$ as ``scaled but jagged'' and $L$ as ``more regular.''

\begin{proposition}[Single-crossing by stakes]\label{prop:sorting_convex}
Maintain Assumptions~\ref{ass:convex_losses} and \ref{ass:gamma_gaps}.
Define the mean and tail terms
\[
m_j\equiv \E[\sigma_j^2(x)],
\qquad
v_j\equiv \E[\sigma_j^4(x)].
\]
Suppose provider $H$ has \emph{better average} but \emph{worse tail} in the sense that
\begin{equation}\label{eq:mean_tail_order}
m_H<m_L
\qquad\text{and}\qquad
v_H>v_L.
\end{equation}
Then there exists a unique cutoff
\begin{equation}\label{eq:qstar_sorting}
q^*\;\equiv\;\phi\,\frac{v_H-v_L}{m_L-m_H}>0
\end{equation}
such that:
\begin{enumerate}[label=(\roman*)]
\item if $q>q^*$ (low stakes), then $\bar U_H(q)>\bar U_L(q)$;
\item if $q<q^*$ (high stakes), then $\bar U_H(q)<\bar U_L(q)$.
\end{enumerate}
\end{proposition}

\begin{proof}
From \eqref{eq:Ubar_convex}, the difference is
\[
\bar U_H(q)-\bar U_L(q)
=
\frac{m_L-m_H}{q}-\phi\,\frac{v_H-v_L}{q^2}.
\]
Under \eqref{eq:mean_tail_order}, the first term is positive and the second term is negative.
Setting the difference to zero yields the unique positive solution \eqref{eq:qstar_sorting}.
The sign comparison for $q$ above or below $q^*$ follows immediately.
\end{proof}

\noindent 
When $q$ is large (errors are cheap), the quadratic term in \eqref{eq:Ubar_convex} is negligible, and users behave close to risk-neutral: they choose the provider with better \emph{mean} performance (here, $H$).
When $q$ is small (errors are expensive), the quadratic term is magnified as $1/q^2$ and tail risk dominates: users prefer the provider whose reliability profile puts less mass on extreme failures (here, $L$).
Lemma~\ref{lem:moments_gamma} shows why regularity matters in high stakes: increasing $\CV$ raises $v$ sharply, and the inspection paradox amplifies this effect by overweighting long gaps.

Regularity reduces the high-variance tail \emph{ex ante} by shrinking the dispersion in gaps.
Calibration (Section~\ref{sec:calibration}) reduces its welfare impact \emph{ex post} by screening out the tail.
Formally, for a calibrated user facing normalised variance $Z$ and reliability index $R$,
\[
U_C(R)-\bar U_B(R)=\E\!\left[\left(\frac{Z}{R}-1\right)_+\right]
=\frac{1}{R}\E[(Z-R)_+],
\]
so the incremental value of calibration is the expected tail loss avoided.
Since $(z-R)_+$ is convex in $z$, any \emph{mean-preserving} increase in dispersion of $Z$ (a convex-order increase)
raises $\E[(Z-R)_+]$ and therefore increases the value of calibration at fixed $R$.
This provides a clean formal sense in which more jagged reliability profiles make calibration/abstention interfaces more valuable.
Thus, providers who are ``scaled but jagged'' can partially substitute for low regularity by investing in
calibration/abstention interfaces, which is precisely what high-stakes users value.

Proposition~\ref{prop:sorting_convex} describes a simple segmentation logic that matches observed practice:
a mass-market product can rationally prioritise scale and average capability, while a professional-grade product can rationally prioritise a more regular reliability profile and calibration features that mitigate tail exposure. The inspection paradox is what makes that differentiation economically salient: it turns rare long-gap failures into a large part of lived experience.

\subsection{Organisational Adoption versus Worker Experience}\label{sec:org_adoption}

Organisations adopt AI through procurement, pilots, and benchmark-based evaluations. Workers then use the system on their actual tasks. AJI predicts a systematic wedge between these two perspectives: evaluations are often closer to gap-uniform sampling, while worker experience is length-biased (Proposition~\ref{prop:inspection_paradox}). The result is a predictable pattern: adoption justified by strong benchmark averages, followed by downstream frustration, workarounds, and distrust when ``surprising failures'' occur. This subsection formalises that wedge and relates it to scaling, calibration, and organisational policy.

Fix an environment with coverage intensity $\lambda$ and gap dispersion $\CV$ as in
Section~\ref{sec:complements}.
Let the organisation evaluate the system on a benchmark that effectively weights gaps uniformly.
Under the Brownian bridge structure, benchmark evaluation corresponds to the mean-variance
$\E_{\text{gap-unif}}[\sigma^2]=\E[X]/6=1/(6\lambda)$.
Define the corresponding \emph{benchmark reliability ratio}
\begin{equation}\label{eq:R_bench_org}
R^{\mathrm{bench}}\;\equiv\;\frac{q}{\E_{\text{gap-unif}}[\sigma^2]}=6\lambda q.
\end{equation}
By contrast, workers face tasks drawn from the arrival distribution and therefore experience the length-biased mean
variance $\E[\sigma^2]=\E[X^*]/6=(1+\CV^2)/(6\lambda)$ (Proposition~\ref{prop:experienced_gap_mean}).
Define the \emph{experienced reliability ratio}
\begin{equation}\label{eq:R_exp_org}
R^{\mathrm{exp}}\;\equiv\;\frac{q}{\E[\sigma^2]}
=\frac{6\lambda q}{1+\CV^2}
=\frac{R^{\mathrm{bench}}}{1+\CV^2},
\end{equation}
which coincides with $R$ in Section~\ref{sec:model} under Poisson gaps ($\CV=1$).

\begin{proposition}[Over-adoption due to the inspection paradox]\label{prop:over_adoption}
Suppose an organisation adopts if the benchmark-based blind value is nonnegative:
$1-1/R^{\mathrm{bench}}\ge 0$, i.e.\ $R^{\mathrm{bench}}\ge 1$.
Workers' experienced blind value is positive if and only if $R^{\mathrm{exp}}\ge 1$,
equivalently $R^{\mathrm{bench}}\ge 1+\CV^2$.
Hence for every $\CV>0$ there is an interval
\[
R^{\mathrm{bench}}\in[1,\,1+\CV^2)
\]
in which the organisation adopts, but workers experience negative expected value from blind reliance.
Under Poisson gaps ($\CV=1$), this over-adoption interval is $R^{\mathrm{bench}}\in[1,2)$.
\end{proposition}

\begin{proof}
By definition, worker value is positive if and only if $R^{\mathrm{exp}}\ge 1$.
Using \eqref{eq:R_exp_org} gives $R^{\mathrm{exp}}\ge 1\iff R^{\mathrm{bench}}\ge 1+\CV^2$.
The benchmark adoption rule requires only $R^{\mathrm{bench}}\ge 1$, so the mismatch interval is
$[1,1+\CV^2)$ for $\CV>0$.
\end{proof}

\noindent 
In the over-adoption region, the organisation is ``right'' under its measurement regime: average benchmark performance is high enough to clear the knife-edge. Workers are also ``right'': the tasks they actually see are drawn from the length-biased distribution and include disproportionately many long-gap tasks where the system fails. AJI, therefore, predicts not random disappointment but systematic disappointment concentrated in exactly the settings where adoption is marginal under benchmark averages.

Proposition~\ref{prop:over_adoption} also clarifies why measured ``returns to scaling'' can diverge across audiences. Differentiating \eqref{eq:R_bench_org}--\eqref{eq:R_exp_org} with respect to $\lambda$ yields
\[
\frac{\partial R^{\mathrm{bench}}}{\partial \lambda}=6q,
\qquad
\frac{\partial R^{\mathrm{exp}}}{\partial \lambda}=\frac{6q}{1+\CV^2}.
\]
Thus, a fixed increase in scale that looks like a $\Delta R^{\mathrm{bench}}$ improvement in benchmark terms delivers
only $\Delta R^{\mathrm{bench}}/(1+\CV^2)$ in lived reliability.
Under Poisson gaps, this is a factor of $1/2$.
AJI therefore predicts that benchmark-based scaling curves can systematically overstate the welfare returns to scaling in irregular domains unless they explicitly incorporate length bias or tail metrics.

If an organisation knows (or can bound) $\CV$, it can align adoption with worker welfare by adopting only when
\begin{equation}\label{eq:corrected_threshold}
R^{\mathrm{bench}}\ge 1+\CV^2.
\end{equation}
The challenge is that $\CV$ is typically not identified by standard leaderboards.
Section~\ref{sec:complements} therefore motivates evaluation protocols that are explicitly usage-weighted or that report tail-risk measures correlated with long gaps.

A common organisational policy is mandatory verification or human review. In the model, verification is valuable precisely because it replicates the calibration benchmark at a cost.

\begin{proposition}[When mandatory verification is welfare-improving]\label{prop:verification_rule}
Suppose workers can verify local reliability at per-task cost $c_v>0$.
Verification reveals $\sigma^2(x)$ and allows the worker to abstain whenever relying would be negative value.
Then the expected payoff from delegation with verification is
\[
U_C(\lambda,q)-c_v
\;\equiv\;
\E\!\left[\left(1-\frac{\sigma^2(x)}{q}\right)_+\right]-c_v,
\]
while blind delegation yields $\bar U_B(\lambda,q)=1-\E[\sigma^2(x)]/q=1-1/R^{\mathrm{exp}}$ (possibly negative),
where $R^{\mathrm{exp}}\equiv q/\E[\sigma^2(x)]$.
Verification strictly increases expected welfare if and only if
\begin{equation}\label{eq:verify_condition}
c_v
<
U_C(\lambda,q)-\bar U_B(\lambda,q)
=
\E\!\left[\left(\frac{\sigma^2(x)}{q}-1\right)_+\right]
=
\E\!\left[\left(\frac{\tilde Z}{R^{\mathrm{exp}}}-1\right)_+\right],
\end{equation}
where $\tilde Z\equiv \sigma^2(x)/\E[\sigma^2(x)]$ has mean $1$.
\end{proposition}

\begin{proof}
Under verification, the worker delegates only when $\sigma^2(x)\le q$, receiving
$(1-\sigma^2(x)/q)_+-c_v$. Taking expectations yields $U_C(\lambda,q)-c_v$.
Without verification, the worker delegates based on expected quality, obtaining
$\bar U_B(\lambda,q)=\E[1-\sigma^2(x)/q]$. The welfare gain from verification is therefore
\[
U_C(\lambda,q)-c_v-\bar U_B(\lambda,q),
\]
and the condition for strict improvement is \eqref{eq:verify_condition}.
\end{proof}

\noindent 
The right-hand side of \eqref{eq:verify_condition} is exactly the expected loss from being forced into the high-error tail. It is larger in more jagged environments (higher dispersion in $Z$) and around the adoption frontier, which is where organisations most often debate rollout. This formalises why ``human-in-the-loop'' policies are most valuable not when the system is terrible or perfect, but when it is \emph{almost} good enough.

Even when adoption is welfare-improving on average, delegation choices can be distorted inside organisations when the
worker who chooses whether to delegate trades off organisational value against a private effort cost of doing the task
without AI. In our reduced-form setup, the misalignment comes from this effort-cost wedge (not from a separate sharing
rule for AI error losses).

\begin{assumption}[Delegation and effort]\label{ass:delegation_effort}
A worker chooses between (i) doing the task themselves at private effort cost $c_e>0$ or (ii) delegating to the AI.
If the worker delegates blindly, the expected incremental value of delegation on a random task is $\bar U_B(R^{\mathrm{exp}})$.
\end{assumption}

\noindent Assumption~\ref{ass:delegation_effort} is intentionally reduced-form: it captures that delegation saves effort but exposes the organisation to error. It abstracts from richer contracting and monitoring to keep the wedge transparent.

\begin{proposition}[Worker over-delegation]\label{prop:over_delegation}
Under Assumption~\ref{ass:delegation_effort}, the organization prefers delegation only when
$\bar U_B(R^{\mathrm{exp}})\ge 0$, i.e.\ $R^{\mathrm{exp}}\ge 1$.
A worker delegates whenever $\bar U_B(R^{\mathrm{exp}})\ge -c_e$.
Hence, workers over-delegate on a nonempty set of parameters whenever $c_e>0$.
Equivalently, in terms of benchmark reliability, the over-delegation region is
\begin{equation}\label{eq:overdelegation_region}
R^{\mathrm{bench}}
\in
\left[
\frac{1+\CV^2}{1+c_e},
\;1+\CV^2
\right),
\end{equation}
with the convention that delegation is never chosen if $\bar U_B<-c_e$.
\end{proposition}

\begin{proof}
Since $\bar U_B(R^{\mathrm{exp}})=1-1/R^{\mathrm{exp}}$, the worker delegates if and only if
$1-1/R^{\mathrm{exp}}\ge -c_e$, i.e.\ $R^{\mathrm{exp}}\ge 1/(1+c_e)$.
Using $R^{\mathrm{exp}}=R^{\mathrm{bench}}/(1+\CV^2)$ from \eqref{eq:R_exp_org} yields
$R^{\mathrm{bench}}\ge (1+\CV^2)/(1+c_e)$.
The organization prefers delegation only if $R^{\mathrm{exp}}\ge 1$, i.e.\ $R^{\mathrm{bench}}\ge 1+\CV^2$.
The difference is \eqref{eq:overdelegation_region}.
\end{proof}

\noindent 
Over-delegation is mechanically larger when the environment is more jagged (larger $\CV$), because the same benchmark reliability ratio translates into lower lived reliability.
This helps explain a common deployment pathology: management sees strong benchmarks and mandates use; workers, facing effort pressure, delegate in marginal regions where expected organisational value is negative; quality incidents follow; trust erodes; and informal norms emerge that restrict use to ``safe'' pockets.

Once workers respond to jaggedness by restricting use to safe regions, their experience becomes endogenous to their delegation policy. Section~\ref{sec:mastery} shows that this creates an additional wedge: scaling improvements can be \emph{economically latent} if they occur outside the region workers visit or can infer about from what they visit (Proposition~\ref{prop:abstention_trap}). Organisational benchmark improvements then overstate worker-experienced improvements twice: first mechanically through \eqref{eq:R_exp_org}, and second dynamically because conservative delegation can prevent the sampling process needed to discover new capabilities.

In sum, AJI is not only a property of model error; it is a property of \emph{institutions and measurement}.
It predicts when benchmark-driven adoption will disappoint, why returns to scaling can look different to managers and workers, and why calibration/verification and regularity investments are disproportionately valuable in exactly the domains where organisations most want reliable AI.

\section{Conclusion}\label{sec:conclusion}

Artificial jagged intelligence (AJI) is a geometric feature of model reliability: performance varies sharply across nearby tasks because knowledge coverage is uneven. In the Brownian--Poisson baseline, this unevenness interacts with the inspection paradox: the gap a typical user encounters is length-biased, so the average experienced error is larger than what is suggested by the average spacing between anchors. This yields simple adoption thresholds for blind users and clarifies why scaling (higher coverage intensity) improves mean reliability while leaving a persistent right tail of failure risk.

A central economic implication is the option value of task-level calibration. When users can observe (or infer) local reliability, they can selectively delegate, turning jaggedness into an exploitable opportunity. Calibration and scale can, therefore, be substitutes (when scale compresses the error distribution and shrinks the tail) or complements (when calibration enables users to safely harvest improvements from scale that would otherwise be unusable).

Finally, mastery matters because calibration is rarely perfect: users must learn where delegation is safe. Under a
Gaussian-process learning model, the worst-case uncertainty about local reliability is bounded above by
$O(\gamma_t/t)$, where $\gamma_t$ captures the effective complexity of the reliability landscape. This creates an
abstention trap: if users only sample tasks they already believe are safe, improvements from scaling can remain
undiscovered. Interface designs that encourage probing near the frontier, provide reliable uncertainty signals, or
make verification cheap can turn otherwise hidden scaling gains into realised welfare improvements.

These results suggest that evaluating AI systems solely by aggregate benchmarks can systematically understate both
risk (via length-biased exposure) and value (via the option value of selective delegation). Understanding AJI is
therefore essential for forecasting adoption, guiding scaling investment, and designing tools that make reliability discoverable rather than accidental.

\newpage

\bibliography{references}
\bibliographystyle{ecta}

\newpage

\appendix

\section{Online Appendix: Generalising the Inspection Paradox}\label{app:inspection}

The main text uses Poisson coverage as a tractable baseline: gaps are exponential, the length-biased gap is
Gamma$(2,\lambda)$ (Proposition~\ref{prop:inspection_paradox}), and several objects admit closed forms.
This appendix records what does \emph{not} rely on the Poisson structure.
The essential point is that the inspection paradox is a manifestation of \emph{size-biased sampling}: a uniformly drawn task location is more likely to fall in larger ``holes'' of the coverage process.
This logic is distribution-free and extends beyond one-dimensional gaps.

\subsection{Disclosure, liability, and evaluation design (brief discussion)}
AJI suggests that policy and market design should focus less on single-number benchmark scores and more on how reliability varies across tasks and how users are exposed to the worst regions.
Two levers are natural: (i) disclosure and interface obligations (calibration, abstention, provenance, and uncertainty
communication) that enable selective delegation, and (ii) evaluation standards that require reporting tails (or
experience-weighted risk) rather than only macro-averages. Liability regimes can also interact with jaggedness: when error costs are externalised, providers may rationally target headline averages while underinvesting in regularity
and tail safety.

\subsection{Length bias for a general gap distribution}\label{app:inspection_general}

Let $\{z_i\}$ denote knowledge points on a line and let $X$ be the (unbiased) gap length between adjacent points.
In the Poisson baseline, $X\sim \mathrm{Exp}(\lambda)$ with $\E[X]=1/\lambda$.
Here we allow $X$ to have an arbitrary distribution with $\E[X]\in(0,\infty)$.

\begin{definition}[Length-biased gap]\label{def:length_bias}
Let $X^*$ denote the length of the gap containing a uniformly drawn task location.
The \emph{length-biased} distribution of $X^*$ is defined by
\begin{equation}\label{eq:length_bias_general_measure}
\Pr(X^*\in A)
\;\equiv\;
\frac{\E\!\left[X\,\1\{X\in A\}\right]}{\E[X]}
\qquad\text{for all measurable }A.
\end{equation}
If $X$ has a density $f_X$, then $X^*$ has density
\begin{equation}\label{eq:length_bias_density}
f_{X^*}(x)=\frac{x f_X(x)}{\E[X]}.
\end{equation}
\end{definition}

\noindent Definition~\ref{def:length_bias} is the formal statement of the inspection paradox:
tasks arrive uniformly in space, not uniformly in \emph{gaps}, so gaps are sampled in proportion to their length.

\begin{lemma}[General inspection-paradox identities]\label{lem:gen_inspection}
If $\E[X]<\infty$, then $X^*$ is well-defined by Definition~\ref{def:length_bias} and for any $m\ge 0$ with
$\E[X^{m+1}]<\infty$,
\begin{equation}\label{eq:length_bias_moments}
\E\!\left[(X^*)^m\right]=\frac{\E[X^{m+1}]}{\E[X]}.
\end{equation}
In particular, if $\E[X^2]<\infty$,
\begin{equation}\label{eq:IP_multiplier_general}
\E[X^*]=\frac{\E[X^2]}{\E[X]}=\E[X]\left(1+\CV^2\right),
\end{equation}
where $\CV\equiv \sqrt{\Var(X)}/\E[X]$.
\end{lemma}

\begin{proof}
By Definition~\ref{def:length_bias},
\[
\E[(X^*)^m]
=\int x^m \frac{x f_X(x)}{\E[X]}\,dx
=\frac{\E[X^{m+1}]}{\E[X]},
\]
and \eqref{eq:IP_multiplier_general} follows by writing
$\E[X^2]=\Var(X)+(\E[X])^2=(1+\CV^2)(\E[X])^2$.
\end{proof}

\noindent 
The Brownian-bridge formula within a gap is unchanged: conditional on the realised gap length $X^*$ and relative position
$t\sim \mathrm{Uniform}(0,1)$ within that gap,
\[
\sigma^2(x)=X^*\,t(1-t).
\]
Hence, whenever $\E[X^2]<\infty$,
\begin{equation}\label{eq:Esig_general_appendix}
\E[\sigma^2(x)]
=
\E[X^*]\E[t(1-t)]
=
\frac{\E[X^*]}{6}
=
\frac{\E[X^2]}{6\,\E[X]}
=
\frac{\E[X]}{6}\left(1+\CV^2\right).
\end{equation}
Equation~\eqref{eq:Esig_general_appendix} is the distribution-free backbone of the paper:
\emph{user-experienced} average uncertainty depends on (i) the \emph{mean} gap length $\E[X]=1/\lambda$ (scale) and
(ii) the \emph{dispersion} of gaps $\CV$ (regularity).

\subsection{Regularity enters as a multiplicative wedge}\label{app:inspection_wedge}

Benchmarks and pilots often resemble \emph{gap-uniform} sampling: they test a fixed number of task families or domains,
implicitly weighting each gap equally.
By contrast, actual use is \emph{length-biased}.
A convenient way to express the resulting measurement wedge is to compare the implied ``reliability ratios.''

Let $\lambda\equiv 1/\E[X]$ denote coverage intensity and fix stakes $q$.
Under gap-uniform sampling, the benchmark-relevant mean variance is $\E_{\text{gap-unif}}[\sigma^2]=\E[X]/6=1/(6\lambda)$.
Under length-biased sampling, \eqref{eq:Esig_general_appendix} gives
$\E_{\text{use}}[\sigma^2]=\E[X](1+\CV^2)/6=(1+\CV^2)/(6\lambda)$.

\begin{proposition}[Benchmark vs.\ experienced reliability]\label{prop:bench_vs_exp_appendix}
Assume $\E[X^2]<\infty$.
Define
\[
R^{\mathrm{bench}}\;\equiv\;\frac{q}{\E_{\text{gap-unif}}[\sigma^2]}
=6\lambda q,
\qquad
R^{\mathrm{exp}}\;\equiv\;\frac{q}{\E_{\text{use}}[\sigma^2]}
=\frac{6\lambda q}{1+\CV^2}.
\]
Then
\begin{equation}\label{eq:R_wedge_appendix}
R^{\mathrm{exp}}=\frac{R^{\mathrm{bench}}}{1+\CV^2}\le R^{\mathrm{bench}},
\end{equation}
with equality if and only if $\CV=0$ (perfectly regular spacing).
For Poisson gaps ($\CV=1$), $R^{\mathrm{exp}}=R^{\mathrm{bench}}/2$.
\end{proposition}

\begin{proof}
Immediate from \eqref{eq:Esig_general_appendix} and the definitions.
\end{proof}

\noindent 
All mean-based ``blind'' results in the main text extend verbatim once reliability is measured in experienced units.
For example, blind expected utility (before truncation at zero) can be written as
\[
\bar U_B
=1-\frac{\E[\sigma^2]}{q}
=1-\frac{1}{R^{\mathrm{exp}}},
\]
so the blind-adoption threshold is $R^{\mathrm{exp}}\ge 1$, i.e.\ $R^{\mathrm{bench}}\ge 1+\CV^2$.
Equation~\eqref{eq:R_wedge_appendix} therefore clarifies why returns to scale can be mismeasured: benchmark improvements map into experienced improvements only after dividing by $(1+\CV^2)$.

\subsection{Examples and tail amplification}\label{app:inspection_examples}

Table~\ref{tab:IP_examples} reports the inspection-paradox multiplier for common gap families.
All ratios are scale-free: they depend on \emph{shape} (regularity/tails), not on $\lambda$.

\begin{table}[t]
\centering
\begin{tabular}{lccc}
\toprule
Gap distribution $X$ & Condition & $\CV(X)$ & $\E[X^*]/\E[X]=1+\CV^2$ \\
\midrule
Deterministic spacing & --- & $0$ & $1$ \\
Exponential (Poisson gaps) & --- & $1$ & $2$ \\
Gamma$(k,\theta)$ (shape $k$, scale $\theta$) & $k>0$ & $1/\sqrt{k}$ & $1+1/k$ \\
Pareto$(\alpha)$ (Type I) & $\alpha>2$ & $\sqrt{\frac{1}{\alpha(\alpha-2)}}$ &
$\frac{(\alpha-1)^2}{\alpha(\alpha-2)}$ \\
\bottomrule
\end{tabular}
\caption{Inspection-paradox multipliers for common gap distributions. For Pareto$(\alpha)$ with $\alpha\le 2$,
$\E[X^2]=\infty$ so $\E[X^*]=\infty$ and the multiplier is unbounded.}
\label{tab:IP_examples}
\end{table}

Lemma~\ref{lem:gen_inspection} requires $\E[X^2]<\infty$ to summarize the paradox with $1+\CV^2$.
If $\E[X]<\infty$ but $\E[X^2]=\infty$ (e.g.\ Pareto with $\alpha\in(1,2]$), then $\E[X^*]=\infty$.
Economically, this means that the user-experienced mean uncertainty $\E[\sigma^2]$ diverges even though a naive ``average gap'' statistic exists. This is an extreme form of jaggedness: rare, enormous holes dominate lived experience.
Most of the paper intentionally avoids this case because mean-based objects like $R$ cease to be informative, but the
qualitative message strengthens: tail behaviour of coverage becomes decisive.

\subsection{Beyond one dimension: size-biased cells}\label{app:inspection_multidim}

In higher-dimensional task spaces, ``gaps'' are not intervals but \emph{sparse regions}.
A natural geometric analogue is the Voronoi tessellation induced by knowledge points: each point owns the region of task space closer to it than to any other point.\footnote{A Voronoi tessellation is a partition of a metric space into regions such that every point in a given region is closer to that region's defining seed point than to any other seed in the set.}

Let $\{C_i\}$ be the Voronoi cells and let $V$ denote the random volume of a \emph{typical} cell.
A uniformly drawn task location falls in cell $i$ with probability proportional to $\mathrm{vol}(C_i)$, so the volume
of the \emph{experienced} cell is size-biased.

\begin{lemma}[Size bias in Voronoi tessellations]\label{lem:voronoi_size_bias}
Let $V$ be the volume of a typical cell and let $V^*$ be the volume of the cell containing a uniformly drawn location.
Then $V^*$ is size-biased:
\[
\Pr(V^*\in A)=\frac{\E\!\left[V\,\1\{V\in A\}\right]}{\E[V]}.
\]
If $\E[V^2]<\infty$,
\[
\E[V^*]=\frac{\E[V^2]}{\E[V]}=\E[V]\left(1+\CV(V)^2\right).
\]
\end{lemma}

\begin{proof}
This is the same argument as in Definition~\ref{def:length_bias} and Lemma~\ref{lem:gen_inspection}, replacing interval
length by cell volume.
\end{proof}

\noindent 
Lemma~\ref{lem:voronoi_size_bias} shows that the inspection-paradox logic survives in any dimension: users are overexposed to large empty regions, measured now by cell volume rather than interval length.
What becomes more complex is mapping cell geometry into the local interpolation risk $\sigma^2(x)$.
In one dimension, Brownian interpolation yields the closed-form bridge variance $\sigma^2=X^*t(1-t)$.
In higher dimensions, the analogue depends on the covariance structure and on the geometry of the nearest anchors.
Nevertheless, the economic primitive persists: size-biasing of ``coverage regions'' creates a systematic wedge between
gap-uniform evaluation and user experience. The exact mapping into a \emph{multiplicative} wedge in $\E[\sigma^2(x)]$
depends on how interpolation risk varies with cell geometry and the covariance structure, so the one-dimensional
$(1+\CV^2)$ factor need not carry over verbatim.

\subsection{Implications for the main text}\label{app:inspection_implications}

This appendix clarifies what Poisson coverage is doing in the main text and what it is not doing.

\begin{itemize}[leftmargin=1.25em]
\item \textbf{Poisson is a convenience, not a knife-edge.}
The inspection paradox is distribution-free: it is length bias (Definition~\ref{def:length_bias}).
Poisson simply makes $X^*$ Gamma$(2,\lambda)$ and yields closed forms for calibrated objects.

\item \textbf{Regularity is a second capability primitive.}
Equation~\eqref{eq:Esig_general_appendix} shows that experienced mean uncertainty decomposes into
scale ($\lambda$) and dispersion ($\CV$).
This is the formal reason that investments that ``close the worst gaps'' can have high welfare return even when they have
little effect on benchmark averages.

\item \textbf{Returns to scale are systematically mismeasured when $\CV>0$.}
Proposition~\ref{prop:bench_vs_exp_appendix} implies that benchmark improvements in $R^{\mathrm{bench}}$ translate into
experienced improvements in $R^{\mathrm{exp}}$ only after dividing by $(1+\CV^2)$.
Under Poisson gaps, this factor is $2$; under heavier tails, it can be much larger.

\item \textbf{Where Poisson matters.}
Objects that depend on the \emph{full distribution} of local variance (e.g.\ calibrated use shares, reasoning cutoffs,
tail losses) depend on the shape of $X^*$, not just its mean.
Poisson provides a clean, closed-form baseline; moving beyond it changes quantitative expressions but preserves the
qualitative logic that jaggedness generates option value and tail-risk wedges.

\item \textbf{Higher-dimensional task spaces retain the same core force.}
In $d>1$ dimensions, intervals are replaced by Voronoi cells, and length bias becomes size bias
(Lemma~\ref{lem:voronoi_size_bias}). The empirical question becomes: how dispersed are coverage regions, and how does that dispersion translate into interpolation risk for the relevant task representation?
\end{itemize}
In short, the inspection paradox is not a statistical curiosity tied to Poisson processes. It is the mechanism that converts irregular coverage into a systematic wedge between what is easy to measure and what is actually experienced, and it is this wedge that drives the paper's results on adoption, scaling incentives, and design.

\section{Online Appendix: Benchmarks in practice and AJI adjustments}\label{app:benchmarks_aji}

This appendix catalogues widely used LLM benchmarks and shows how their headline scores can be
\emph{systematically} reinterpreted and, where possible, \emph{adjusted} in light of Artificial Jagged Intelligence (AJI).
The main paper's mechanism is the inspection paradox: users do not encounter ``domains'' or ``gaps'' uniformly; they encounter them in proportion to how much of the task space those regions occupy. This creates a
measurement wedge between (i) a benchmark that macro-averages across a curated set of tasks/domains and
(ii) experienced reliability in deployment.

\subsection{A general correction: from macro-averages to experienced reliability}\label{app:benchmarks_general}

Many benchmarks report a \emph{single-number score} by (approximately) averaging performance across a set of
domains/tasks.
In the AJI model, the analogous procedure corresponds to \emph{gap-uniform} evaluation: each region of the task space
is represented once, regardless of how much of the space it occupies.

To make the adjustment transparent, consider a discrete task space of total length $L$ partitioned into $n$ regions
(``gaps'') with lengths $(X_i)_{i=1}^n$ and $\sum_i X_i=L$.  Within region $i$, the within-gap \emph{average} posterior
variance is $X_i/6$ (from $\sigma^2(x)=X_i t(1-t)$ and $\int_0^1 t(1-t)\,dt=1/6$).  Then:
\begin{align}
\E_{\text{bench}}[\sigma^2]
&=\frac{1}{n}\sum_{i=1}^n \frac{X_i}{6}
=\frac{L}{6n},
\\
\E_{\text{user}}[\sigma^2]
&=\sum_{i=1}^n \Pr(\text{user in }i)\cdot \E[\sigma^2\mid i]
=\sum_{i=1}^n \frac{X_i}{L}\cdot \frac{X_i}{6}
=\frac{1}{6L}\sum_{i=1}^n X_i^2.
\end{align}
The ratio is
\begin{equation}\label{eq:inspection_adjustment_discrete}
\frac{\E_{\text{user}}[\sigma^2]}{\E_{\text{bench}}[\sigma^2]}
=
\frac{n\sum_{i=1}^n X_i^2}{\left(\sum_{i=1}^n X_i\right)^2}
=
1+\CV_X^2,
\end{equation}
where $\CV_X\equiv \sqrt{\Var(X)}/\E[X]$ is the coefficient of variation of gap lengths.
In the Poisson baseline $\CV_X=1$, so a gap-uniform benchmark is \emph{twice as optimistic} about average reliability:
$\E_{\text{user}}[\sigma^2]=2\,\E_{\text{bench}}[\sigma^2]$.

\paragraph{What can be adjusted in practice?}
Equation~\eqref{eq:inspection_adjustment_discrete} suggests two empirically implementable benchmark adjustments.

\begin{enumerate}[label=(\roman*),leftmargin=1.25em]
\item \textbf{Exposure reweighting (preferred when feasible).}
Let $d$ index benchmark domains (subjects, tasks, scenarios).  If the benchmark reports a per-domain score
$s_d(m)$ for model $m$ and aggregates with weights $\omega_d^{\text{bench}}$ (often uniform), then a deployment-aligned score is
\begin{equation}\label{eq:deployment_weighted_score}
S^{\text{dep}}(m)=\sum_d \omega_d^{\text{dep}}\, s_d(m),
\qquad \sum_d \omega_d^{\text{dep}}=1,
\end{equation}
where $\omega_d^{\text{dep}}$ is estimated from usage logs, product telemetry, or other prompt-frequency data.

\item \textbf{Inspection-paradox ``jaggedness penalty'' (diagnostic when exposure weights are unavailable).}
If only an unweighted or uniform-weight benchmark is available, one can report a dispersion statistic across domains,
such as $\CV_d(e_d)$ for domain-level error rates $e_d\equiv 1-s_d$.
In settings where domain size and domain error are positively associated (the AJI mechanism),
$\CV_d(e_d)$ is informative about how far the macro-average can drift from experienced reliability.
This is not a structural estimator without additional assumptions, but it is a practical diagnostic: high dispersion signals a potentially large inspection wedge.
\end{enumerate}

\subsection{Benchmark-by-benchmark: design, use, and AJI adjustments}\label{app:benchmarks_catalogue}

Below, we summarise prominent benchmarks and indicate what an AJI-style adjustment would require.

\paragraph{GLUE and SuperGLUE (macro-averaged NLU suites).}
GLUE and SuperGLUE were early ``single-number'' NLU benchmarks designed to summarise progress across multiple tasks.
Both aggregate task metrics via an (approximately) \emph{unweighted average across tasks}, and when a task has multiple
metrics, the task score is an average of those metrics. \citep{wang2018glue,wang2019superglue}

\emph{AJI implication.}  Task-uniform aggregation corresponds to gap-uniform weighting (each block gets equal weight); \eqref{eq:inspection_adjustment_discrete} shows the resulting wedge relative to usage-weighted experience.  If an organisation's actual workload concentrates in a subset of tasks
(e.g.\ entailment-like tasks), GLUE-style macro-averages can misstate experienced performance.

\emph{Adjustment.}  (i) Replace task-uniform weights with deployment weights $\omega_d^{\text{dep}}$ in
\eqref{eq:deployment_weighted_score}. (ii) Report dispersion across tasks (e.g.\ $\CV_d(e_d)$ or the 10th percentile
task score) to surface jaggedness.

\paragraph{MMLU (broad subject test of knowledge/reasoning).}
MMLU evaluates multiple-choice performance over 57 subjects and reports an overall accuracy aggregated across tasks
and examples. \citep{hendrycks2021measuring}

\emph{AJI implication.}  MMLU's overall score is a convenient proxy for \emph{mean} capability, but it does not reveal
how uneven performance is across subjects (the AJI ``holes''), nor does it weight subjects by how frequently they are
encountered in deployment (or by stakes).

\emph{Adjustment.}  (i) When subject-level scores are available, compute deployment weights across subjects
(e.g.\ enterprise vs.\ consumer mixes) and form $S^{\text{dep}}$ in \eqref{eq:deployment_weighted_score}. (ii)
Complement the mean with tail summaries (worst-quintile subjects, dispersion of errors) to approximate the inspection wedge.

\paragraph{BIG-bench (long-tail capability suite).}
BIG-bench contains 204 tasks and defines aggregate performance as the \emph{average across tasks} of each task's
``normalised'' preferred metric. \citep{srivastava2023bigbench}

\emph{AJI implication.}  Because BIG-bench deliberately spans many rare/long-tail tasks, it is naturally interpreted as
probing the \emph{right tail} of difficulty rather than estimating a usage-weighted average.  Under AJI, this is
useful: the tail is precisely where surprising failures and adoption frictions come from.

\emph{Adjustment.}  (i) Treat BIG-bench as a tail-risk module and report the distribution of task scores (e.g.\ quantiles), not only the average. (ii) For a deployment-aligned level, combine it with a usage-weighted benchmark using a mixture:
\[
S^{\text{mix}}(m)
=(1-\pi)\,S^{\text{main}}(m)+\pi\,S^{\text{tail}}(m),
\]
where $\pi$ is the (estimated) share of ``tail'' queries in the deployment environment.

\paragraph{HELM (scenario $\times$ metric matrix).}
HELM evaluates language models across many scenarios and multiple desiderata (beyond accuracy), explicitly organising
evaluation as a scenario-by-metric matrix. \citep{liang2022holistic}

\emph{AJI implication.}  HELM is well-suited to AJI because it already emphasises \emph{coverage} (many scenarios) and
\emph{multi-metric tradeoffs}.  The remaining issue is aggregation: deployment rarely values scenarios uniformly.
Moreover, AJI stresses that the stakes vary by scenario, so a welfare-aligned aggregation should incorporate
both \emph{frequency} and \emph{stakes}.

\emph{Adjustment.}  (i) Construct a deployment-weighted scenario score using \eqref{eq:deployment_weighted_score}.  (ii)
If stakes differ substantially across scenarios, weight scenarios by an estimate of $q_d^{-1}$ (or by policy-relevant
loss weights) to reflect that errors in some scenarios are much more costly.

\paragraph{TruthfulQA (adversarial truthfulness/misconception traps).}
TruthfulQA is designed to test whether models produce truthful answers in settings where humans often repeat common
misconceptions; it comprises 817 questions across 38 categories. \citep{lin2022truthfulqa}

\emph{AJI implication.}  TruthfulQA is a targeted probe of an important failure mode rather than a representative
average-case workload. Under the AJI lens, it is best treated as measuring the probability mass and overshoot of a
high-cost tail (a ``catastrophic'' region), which maps closely to the cost-of-blindness object in the main paper.

\emph{Adjustment.}  Use TruthfulQA (and similar safety/truthfulness suites) as an estimate of tail failure risk and
report it alongside an average-case benchmark:
\[
\text{Report:}\quad (S^{\text{avg}}(m),\; \Pr(\text{tail failure}\mid m)).
\]
This mirrors the AJI prescription in structure: pair an average-case metric with a tail-risk (or coverage) metric.
These benchmark objects are not meant as direct estimates of $(U_C,s_C)$, which depend on an endogenous abstention policy
and the usage distribution.

\paragraph{MT-Bench and Chatbot Arena (interactive/chat evaluation).}
MT-Bench is a fixed multi-turn question set scored by an LLM-judge methodology. \citep{zheng2023judging}
Chatbot Arena collects pairwise human preferences over model outputs on crowdsourced prompts and fits a statistical
preference model to produce a ranking/score. \citep{chiang2024chatbotarena}

\emph{AJI implication.}  MT-Bench is closer to a conventional gap-uniform benchmark: it is a small, fixed set of prompts
and thus inherits the usual representativeness concern.
Chatbot Arena is conceptually closer to \emph{user-weighted} sampling because prompts are contributed by users, which
partially internalises the inspection-paradox logic: more frequent query types are more likely to appear.

\emph{Adjustment.}
\begin{itemize}[leftmargin=1.25em]
\item For MT-Bench: stratify the prompts by domain (or inferred topic) and apply deployment weights, or expand the prompt set using product-telemetry-derived sampling so that prompt frequency matches deployment.
\item For Chatbot Arena: reweight (or subsample) prompts to match the target user population and use-case mix (enterprise vs.\ consumer, language mix, safety-critical vs.\ casual).  Arena provides an \emph{estimable} starting point for $\omega_d^{\text{dep}}$ in \eqref{eq:deployment_weighted_score}, but the weights are environment-specific.
\end{itemize}

\paragraph{A note on ``bias corrections'' in benchmark practice.}
It is increasingly standard to correct benchmarks for known, mechanically induced biases (e.g.\ output-length bias in
automatic preference evaluators). Length-Controlled AlpacaEval provides an example of such a correction strategy in
LLM evaluation. \citep{dubois2024lengthcontrolled}
The AJI adjustments advocated here are analogous in spirit: they aim to correct \emph{sampling/weighting} biases that
systematically overstate experienced reliability by underweighting sparse, failure-prone regions.

\subsection{A practical reporting template (minimal AJI additions)}\label{app:benchmarks_template}

For an evaluator who wants to remain close to existing practice but incorporate AJI, a minimal template is:

\begin{enumerate}[label=(\roman*),leftmargin=1.25em]
\item \textbf{Mean score} (status quo): the existing benchmark aggregate (e.g.\ macro-average).
\item \textbf{Dispersion/tails}: report at least one dispersion statistic across benchmark domains/tasks
(e.g.\ $\CV_d(e_d)$, 10th percentile domain score, or worst-$k$ domains).
\item \textbf{Deployment reweighting (when possible)}: publish a deployment-weighted score
$S^{\text{dep}}$ as in \eqref{eq:deployment_weighted_score} for a small number of canonical environments
(e.g.\ ``consumer chat'', ``enterprise office workflow'', ``coding assistant'').
\item \textbf{Tail-risk module}: report a targeted safety/truthfulness benchmark (e.g.\ TruthfulQA) separately,
rather than folding it into the same macro-average.
\end{enumerate}
\newpage

\begin{table}[h]
\centering
\footnotesize
\caption{Summary: how common benchmark designs map into AJI and what an adjustment requires.}
\label{tab:aji_benchmark_summary}
\begin{tabular}{p{0.17\linewidth} p{0.18\linewidth} p{0.20\linewidth} p{0.25\linewidth} p{0.16\linewidth}}
\toprule
Benchmark & Default aggregation & AJI issue (why the headline can mislead) & AJI-style adjustment (what to do) & Data required \\
\midrule
GLUE / SuperGLUE \citep{wang2018glue,wang2019superglue}
& Mostly task-uniform averages
& Gap-uniform weighting ignores exposure and tails
& Reweight tasks by deployment mix; report dispersion/quantiles
& Task-level scores; usage weights \\
\addlinespace
MMLU \citep{hendrycks2021measuring}
& Overall accuracy over many subjects
& Mean hides uneven subject reliability; subject exposure differs by deployment
& Reweight subjects; publish tail metrics (worst decile subjects)
& Subject scores; domain usage shares \\
\addlinespace
BIG-bench \citep{srivastava2023bigbench}
& Average across 204 tasks (normalised)
& Suite is long-tail by design; average is not ``typical usage''
& Treat as tail module; report score distribution; mix with main benchmark using $\pi$
& Task scores; estimate of tail query share $\pi$ \\
\addlinespace
HELM \citep{liang2022holistic}
& Scenario $\times$ metric matrix; aggregation varies
& Uniform scenario aggregation rarely matches deployment; stakes differ by scenario
& Deployment- and stakes-weight scenario aggregation
& Scenario metrics; scenario usage and stakes weights \\
\addlinespace
TruthfulQA \citep{lin2022truthfulqa}
& Targeted truthfulness failure rate
& Not average-case; measures misconception ``trap'' tail
& Report separately as tail risk; combine via mixture weights when needed
& Tail benchmark score; tail prevalence estimate \\
\addlinespace
MT-Bench \citep{zheng2023judging}
& Fixed prompt set, LLM-judge scoring
& Small fixed set is gap-uniform; may miss worst holes
& Expand/stratify prompts; reweight by deployment prompt distribution
& Prompt metadata; deployment prompt mix \\
\addlinespace
Chatbot Arena \citep{chiang2024chatbotarena}
& Crowdsourced prompts + human pairwise preferences
& Closer to user-weighted sampling, but reflects the Arena user population
& Reweight prompts/users to target deployment; stratify by domain/safety class
& Prompt dataset; target-population weights \\
\bottomrule
\end{tabular}
\end{table}

\end{document}